\documentclass{llncs}

\usepackage{latexsym}
\usepackage{amssymb}
\usepackage{amsfonts}
\usepackage{bbm}
\usepackage{mathrsfs}
\usepackage{alltt}
\usepackage{url}
\usepackage[english]{babel}
\usepackage{graphicx}
\usepackage[latin1]{inputenc}
\usepackage{bussproofs}
\EnableBpAbbreviations	
\usepackage{paralist}
\usepackage[all]{xy}
\usepackage{yhmath}
\usepackage[usenames]{color}
\usepackage{colortbl}
\usepackage{stmaryrd}
\usepackage{framed}
\usepackage{rotating}

\usepackage[font=small,labelfont=bf]{caption}
\usepackage[font=footnotesize]{subfig}

\def \FOL{{\sf FOL}}
\def \FODL{{\sf FODL}}
\def \LTL{{\sf LTL}}
\def \TL{{\sf TL}}

\def \PDL{{\sf PDL}}
\def \DLTL{{\sf DLTL}}

\def \ETBR{{\sf ETBR}}

\def \PRA{{\sf PRA}}
\def \CR{{\sf CR}}
\def \RA{{\sf RA}}

\def \sPFA{{\star{\sf PFA}}}
\def \PFA{{\sf PFA}}
\def \fPFA{{\sf fPFA}}
\def \CFA{{\sf CFA}}
\def \FA{{\sf FA}}

\def \sPFAU{{\star{\sf PFAU}}}
\def \PFAU{{\sf PFAU}}
\def \fPFAU{{\sf fPFAU}}
\def \CFAU{{\sf CFAU}}
\def \FAU{{\sf FAU}}

\def \halfthinspace{\relax\ifmmode\mskip.5\thinmuskip\relax\else\kern.8888em\fi}\let \hts=\halfthinspace

\def \({\left (}
\def \){\right )}
\def \<{\left\langle}
\def \>{\right\rangle}

\def \ajoin{{\hts + \hts}}
\def \aunit{1}
\def \ameet{{\hts \cdot \hts}}
\def \azero{0}
\def \acompl#1{\overline{#1}}
\def \acomple{\mbox{$^{\mbox{--}}$}}
\def \acompo{{\hts ; \hts}}
\def \aid{1^\prime}
\def \aconv#1{\setbox13\hbox{$#1$}\ifdim\wd13<12pt\breve{#1}\else{\(#1\)}\breve{\ }\fi}
\def \aconve{\breve{\ }}
\def \afork{{\hts \nabla \hts}}
\def \aidu{1^\prime_{\sf U}}

\def \uaunitu{{_{\textup{\textsf{U}}}\aunit_{\textup{\textsf{U}}}}}

\def \pjoin{\cup}
\def \punit#1{#1}
\def \pmeet{\cap}
\def \pzero{\emptyset}
\def \pcompl#1{\overline{#1}}
\def \pcomple{\mbox{$^{\mbox{--}}$}}
\def \pcompo{{\hts \circ \hts}}
\def \pid{Id}
\def \pconv#1{\setbox13\hbox{$#1$}\ifdim\wd13<10pt\stackrel{\smile}{#1}\else{\(#1\)}^{\smile}\fi}
\def \pconve{^{\smile}}
\def \pfork{{\hts \underline{\nabla} \hts}}

\def \setof#1#2{\setbox1\hbox{$#1$}
                \setbox2\hbox{$#2$}
                \ifdim \ht1 > \ht2
                   \left \{ \left . \, #1 \, \right \vert \, #2 \, \right \}
                \else
                   \left \{ \, #1 \, \left \vert \, #2 \, \right . \right \} 
                \fi}
\def \set#1{\left \{\, #1 \,\right \}}
\def \pair#1#2{\left \langle #1, #2 \right \rangle}

\def \nat{\mathbb{N}}

\usepackage{xspace}
\usepackage{xcolor}
\usepackage{framed,color}
\usepackage{todonotes}

\newcommand\thmt[4]{\vspace{1em}\noindent {{\bf #1 \ref{#2}} (#3).} \emph{#4}}

\title{On the construction of explosive relation algebras}

\author{Carlos~G.~Lopez~Pombo\inst{1} \and Marcelo~F.~Frias\inst{2} \and Thomas~S.E.~Maibaum\inst{3}}
 \institute{
   Universidad de Buenos Aires, Facultad de Ciencias Exactas y Naturales, Departamento de Computaci\'on and Instituto de Investigaci\'on en Ciencias de la Computaci\'on, CONICET--Universidad de Buenos Aires (ICC).
   \email{clpombo@dc.uba.ar}
   \and
   Department of Software Engineering, Buenos Aires Institute of Technology (ITBA) and Consejo Nacional de Investigaciones Cient\'{\i}ficas y T\'ecnicas (CONICET).
   \email{mfrias@itba.edu.ar}
   \and
   Emeritus Professor, Department of Computing and Software, McMaster University
   \email{tom@maibaum.org}
 }
 
\begin{document}

\maketitle

\begin{abstract}
Fork algebras are an extension of relation algebras obtained by extending the set of logical symbols with a binary operator called \emph{fork}. This class of algebras was introduced by Haeberer and Veloso in the early 90's aiming at enriching relation algebra, an already successful language for program specification, with the capability of expressing some form of parallel computation.

The further study of this class of algebras led to many meaningful results linked to interesting properties of relation algebras such as representability and finite axiomatizability, among others. Also in the 90s, Veloso introduced a subclass of relation algebras that are expansible to fork algebras, admitting a large number of non-isomorphic expansions, referred to as \emph{explosive relation algebras}.

In this work we discuss some general techniques for constructing algebras of this type.
\end{abstract}

\section{introduction}
\label{introduction}

A \emph{relation algebra} is an algebraic structure formed by three relational constants understood as the empty, universal and identity relations, typically represented by the symbols ``$\azero$'', ``$\aunit$'' and ``$\aid$'', respectively; two unary operators playing the role of complement of a relation (with respect to the universal relation) and transposition (or converse) of a relation, typically represented by the symbols ``$\acomple$'' and ``$\aconve$'', respectively; and binary operators for product, co-product and relative product (also commonly referred to as composition) of relations, typically represented by symbols ``$\ameet$'', ``$\ajoin$'' and ``$\acompo$'', respectively.

In \cite{tarski:jsl-6_3}, Tarski noted that the calculus of (binary) relations ``[...] has had a strange and rather capricious line of historical development.'', but leaving historical discussions aside, it is fair to consider that it was proposed, still not properly presented and formalised, by him in the previously cited work. 
There, Tarski 
committed himself to the development of the \emph{calculus of relations} ($\CR$). In the first place, Tarski introduces the \emph{elementary theory of binary relations} ($\ETBR$), as a logical formalisation of the algebras of binary relations in a kind of definitional extension of first-order logic (see \cite[Axs.~1.--12.]{tarski:jsl-6_3}); then, the calculus of relations can be obtained from the elementary theory of binary relations by restricting the language to sentences without individual variables (see \cite[Thms.~\texttt{1}--\texttt{15}]{tarski:jsl-6_3}\footnote{Theorems~\texttt{I}--\texttt{VII}, originally due to Huntington \cite[\S1]{huntington:tams-5_3}, provide a characterisation of the meaning of the absolute constants (i.e. the Boolean fragment of the logic), and Thms.~\texttt{8}--\texttt{15} express the fundamental properties of the relative ones.}). While the modern equational formulation of the calculus of relations is not explicit in Tarski's work, it is hinted at after the proof that \cite[Thm.~\texttt{32}]{tarski:jsl-6_3} (originally proved by Shr\"{o}der in \cite[1), pp.~150--153]{schroder1895}) follows from Axioms~\texttt{1} to~\texttt{15}, by proving a general metalogical result stating that any sentence of the calculus of relations can be transformed into an equivalent sentence of the form $R = S$.

At the end Tarski states five questions related to the calculus of relations, its class of models and the \emph{algebras of binary relations}\footnote{Tarski refers to these structures as ``a class of binary relations which contains $1$, $0$, $1'$, $0'$ and is closed under all the operations considered in the calculus [of relations]'', without providing a proper name for such an intended class of models.}, with three of them of particular interest to this work:
\begin{itemize}
\item Is every model of the calculus of relations isomorphic to an algebra of binary relations?
\item Is it true that every formula that is valid in every algebra of binary relations is provable in the calculus of relations? 
\item Is it true that every formula of the elementary theory of binary relations can be transformed into an equivalent formula of the calculus of relations?
\end{itemize}

It was Lyndon in \cite{lyndon:ams2-51_2} who gave a negative answer to the first two questions by exhibiting a finite, non-simple and non-trivial algebra of relations that is not representable as an algebra of binary relations. After that, it was Monk in \cite{monk:mmj-11} who proved that the class of the algebras of binary relations cannot be finitely axiomatised. The third question was answered negatively by Tarski (hinted at in op. cit., pp.~88--89) by making more precise the existence of \emph{uncondensable}\footnote{To \emph{condense} a formula, as used by L\"{o}wenheim following Schr\"{o}der's terminology \cite[pp.~550]{schroder1895}, is to transform a formula of ${\sf ETBR}$ into another one in which no quantifiers or individual variables appear.} formulae proved by Korselt (and published in \cite[Thm.~1]{lowenheim:ma-76}). Tarski's detailed proof of the equipolence of the calculus of relations and the three variable fragment of the dyadic first-order predicate logic appeared for the first time in a book manuscript \cite{tarski:unp43-45}, and later was published in \cite[\S3.9]{tarski87}. By mid-50's Tarski had already adopted a completely equational presentation for relation algebras (see \cite[pp.~60]{tarski:im-17} and \cite[\S3]{tarski:im-18}).

The \emph{fork algebras} ($\FA$) were introduced by Armando Haeberer and Paulo~A.~S. Veloso in \cite{haeberer:ifiptc291} as extensions of relation algebras, obtained by adding a new operator called \emph{fork} (typically represented as ``$\afork$''). They arose in the search for a formalism suitable for software specification and verification. In \cite[Chap.~3, pp.~20]{frias02} Frias gave a detailed discussion of the evolution of fork algebras, focussing the reader's attention to the concepts behind such specific direction. The interpretation of $\afork$ is defined by the following first-order formula: given relations $R$ and $S$,
$$
R \afork S = \setof{\pair{x}{y}}{(\exists y_1, y_2 \in U)(\pair{x}{y_1} \in R \land \pair{x}{y_2} \in S \land y = y_1 \star y_2}
$$
\noindent where $\star: U^2 \to U$ is an injective function acting as an encoding of pairs of elements $U$, the set over which the relations are defined.\\

The class of fork algebras have some particularly attractive features:
\begin{itemize}
\item every fork algebra is isomorphic to an algebra whose domain is a set of binary relations (Frias et al. in \cite{frias:fi-32} and, independently, Gyuris in \cite{gyuris:tcs-188_1_2}),
\item it has a finite equational calculus (Frias et al. in \cite{frias:jigpl-5_3}),
\item it has expressive power capable of providing an interpretation language for many logics. Given a logic $\mathcal{L}$, an interpretation is a relational algebraization of $\mathcal{L}$. This is done by resorting to a semantics-preserving mapping $T_{\mathcal{L}} : \mathit{Formulas}_{\mathcal{L}} \to \mathit{RelDes(X)}$ for some set of relational variables $X$, translating $\mathcal{L}$-formulas to relational terms. 
%
Some known interpretability results are: first-order predicate logic ($\FOL$) in fork algebra \cite{frias02}, $\PDL$ in fork algebra \cite{frias:jancl-8}, first-order dynamic logic ($\FODL$) in fork algebra \cite{frias:relmics01+}, $\LTL$, $\TL$ \cite{frias:relmics03} and their first-order versions in fork algebra \cite{frias:jlap-66_2}, and propositional dynamic linear temporal logic ($\DLTL$) in fork algebra \cite{frias:relmics05}, among others.
\end{itemize}

The existence of a representability theorem and a finetely axiomatizable complete calculus for the class of fork algebras motivates the study of the subclass of relation algebras obtained by taking the relational reduct of the fork algebras, resulting also in a subclass of the proper relation algebras.\\

In this paper we present some techniques for constructing relation algebras admitting a large amount of non-isomorphic expansions to a fork algebra. This class of algebras was introduced by Veloso in \cite{veloso:05-96,veloso:ES-418-96} and because they possess this property, they are called explosive. The definitions and results in this work are strongly inspired by the reports mentioned above and joint technical discussion with Paulo A.S. Veloso.

\section{Preliminaries}
\label{preliminaries}
In this section we fix notation and present definitions and results used in the rest of the paper. In general, we adopt the algebraic notation used in \cite{burris81}, but resorting to the symbols proposed by Tarski in \cite{tarski:jsl-6_3} and used in \cite{jonsson:ajm-73,jonsson:ajm-74}. In general, axioms, deduction rules and proofs will follow the notation used in \cite{enderton72}.\\

As we mentioned before, Tarski's development of relation algebras started by introducing the elementary theory of binary relations as a definitional extension of first-order logic with the relational operators proposed by Schr\"oder in \cite{schroder1895} and then moved on to the calculus of relations by restricting the language to formulae stating properties of relational terms exclusively. The following definitions introduce the \emph{calculus of relations}.

\begin{definition}[Formulae of the calculus of relations]
Let $\mathcal{R}$ be a set of relation variables, then the set of relation designations is the smallest set $\mathit{RelDes}(\mathcal{R})$ such that:
\begin{itemize}
\item $\mathcal{R} \cup \set{\aunit, \azero, \aid} \subseteq
  \mathit{RelDes}(\mathcal{R})$, 
\item If $r, s \in \mathit{RelDes}(\mathcal{R})$, then $\set{r \ajoin s, r
  \ameet s, \acompl{r}, r \acompo s, \aconv{r}} \subseteq
  \mathit{RelDes}(\mathcal{R})$.
\end{itemize}
Then the set of formulae of $\CR$ is the smallest set $\mathit{CRForm}(\mathcal{R})$ such that:
\begin{itemize}
\item If $r, s \in \mathit{RelDes}(\mathcal{R})$, then $r = s \in \mathit{CRForm}(\mathcal{R})$, 
\item If $f, g \in \mathit{CRForm}(\mathcal{R})$, then $\set{\neg f, f \lor g} \subseteq \mathit{CRForm}(\mathcal{R})$.
\end{itemize}
\end{definition}

The remaining propositional connectives can be defined as usual in terms of negation ($\neg$) and disjunction ($\lor$).

\begin{definition}[The calculus of relations, \cite{tarski:jsl-6_3}, pp.~76--77]
\label{def_calculus-of-relations}
Let $\mathcal{R}$ be a set of relation variables, then $\CR$ is defined for the formulae in $\mathit{CRForm}(\mathcal{R})$ by\footnote{Tarski presented $\CR$ incorporating axioms in order to characterize the relative addition ($\mbox{\c+}$) and the diversity ($0'$) operators. Tarski's axioms for these two operators are:
$$
\begin{array}{l}
r \; \mbox{\c+} \; s = \acompl{\acompl{r} \acompo \acompl{s}}\\
0' = \acompl{\aid}
\end{array}
$$}:
\begin{itemize}
\item the axioms for the Boolean operators, and
\item the following axioms for the relational operators:
$$
\begin{array}{l}
(r = s \land r = t) \Longrightarrow s = t\ ,\\
r = s \Longrightarrow (r \ajoin t = s \ajoin t \land r \ameet t = s
   \ameet t)\ ,\\
r \ajoin s = s \ajoin r \land r \ameet s = s \ameet r\ ,\\
r \ajoin (s \ameet t) = (r \ajoin s) \ameet (r \ajoin t) \land r
   \ameet (s \ajoin t) = (r \ameet s) \ajoin (r \ameet t)\ ,\\
r + \azero = r \land r \ameet \aunit = r\ ,\\
r \ajoin \acompl{r} = \aunit \land r \ameet \acompl{r} = \azero\ ,\\
\acompl{\aunit} = \azero\ ,\\
\aconv{\aconv{r}} = r\ ,\\
\aconv{r \acompo s} = \aconv{s} \acompo \aconv{r}\ ,\\
r \acompo (s \acompo t) = (r \acompo s) \acompo t\ ,\\
r \acompo \aid = r\ ,\\
r \acompo \aunit = \aunit \lor \aunit \acompo \acompl{r} = \aunit\ ,\\
(r \acompo s) \ameet \aconv{t} = 0 \Longrightarrow (s \acompo t)
   \ameet \aconv{r} = 0\ .\\
\end{array}
$$
\end{itemize}
\end{definition}

While Tarski did not commit to any set of inference rules for structuring deduction, one can assume any appropriate set for the boolean operators (``$\neg$'' and ``$\lor$''), and the equality (``$=$'').
%

As we mentioned in the introduction, in Tarski's presentation of the representation problem \cite[pp.~88]{tarski:jsl-6_3} there is only an implicit definition of the intended models of the calculus of relations. For the purpose of the present work, we adopt the formal definition given by J{\'{o}}nsson and Tarski in \cite{jonsson:ajm-74}.

\begin{definition}[Proper relation algebras, \cite{jonsson:ajm-74}, Def.~4.23]
\label{def_proper-relation-algebra}
A \emph{proper relation algebra} is an algebraic structure $\< A, \pjoin, \pmeet, \pcomple, \pzero, \punit{E}, \pcompo, \pconve, \pid \>$ in which $A$ is a set of binary relations on a set $U$, $\pjoin$, $\pmeet$ and $\pcompo$ are binary operations, $\pcomple$ and $\pconve$ are unary operations and $\pzero$, $\punit{E}$ and $\pid$ are distinguished elements of $A$ satisfying:
\begin{itemize}
\item $A$ is closed under $\pjoin$ (i.e. set union),
\item $A$ is closed under $\pmeet$ (i.e. set intersection),
\item $A$ is closed under $\pcomple$ (i.e. set complement with respect
  to $\punit{E}$),
\item $\pzero \in A$ is the empty relation on the set $U$,
\item $\punit{E} \in A$ and $\bigcup_{r \in A} r \subseteq \punit{E}$,
\item $A$ is closed under $\pcompo$, defined as follows
$$x \pcompo y  = \setof{\pair{a}{b} \in U \times U}{ (\exists c)
  (\pair{a}{c} \in x \land \pair{c}{b} \in y)}\ ,$$
\item $A$ is closed under $\pconve$, defined as follows
$$\pconv{x} = \setof{\pair{a}{b} \in U \times U}{\pair{b}{a} \in x}\
  ,$$ 
\item $\pid \in A$ is the identity relation on the set $U$.
\end{itemize} 
The class of proper fork algebras will be denoted as $\PRA$.
\end{definition}

In \cite{jonsson:ajm-74}, J{\'{o}}nsson and Tarski proved that the axiom $r \acompo \aunit = \aunit \lor \aunit \acompo \acompl{r} = \aunit$ forces the models to be simple, a property that is not necessarily satisfied by the proper relation algebras, so that is why their equational presentation of $\CR$ does not include it.

\begin{definition}[Equational formulae of the relational calculus]\ \\
\label{eqcrform}
Let $\mathcal{R}$ be a set of relation variables, then the set of formulas of $\CR$ is the set $\setof{r = s}{r, s \in \mathit{CRForm}(\mathcal{R})}$.
\end{definition}

\begin{definition}[The equational calculus of relations, \cite{jonsson:ajm-74}, Def.~4.1]
\label{def_new-calculus-of-relations}
Let $\mathcal{R}$ be a set of relation variables, then $\CR$ is defined for the formulae in $\mathit{CRForm}(\mathcal{R})$ by:
\begin{itemize}
\item the axioms for the Boolean operators, and
\item the following axioms for the relational operators\footnote{The last axiom, known as the Dedekind formula, is equivalent to $r \acompo s \ameet t = 0 \ \mbox{iff} \ t \acompo \aconv{s} \ameet r = 0 \ \mbox{iff} \ \aconv{r} \acompo t \ameet s = 0$ known as \emph{cycle rule}.}: for all $r, s, t \in A$
$$
\begin{array}{l}
r \acompo (s \acompo t) = (r \acompo s) \acompo t \\ 
(r \ajoin s) \acompo t = (r \acompo t) \ajoin (s \acompo t) \\
\aconv{r \ajoin s} = \aconv{r} \ajoin \aconv{s} \\ 
\aconv{\aconv{r}} = r \\
r \acompo \aid = r \\
\aconv{r \acompo s} = \aconv{s} \acompo \aconv{r} \\ 
(r \acompo s) \ameet t \leq (r \ameet (t \acompo \aconv{s})) \acompo (s \ameet (\aconv{r} \acompo t))
\end{array}
$$
\end{itemize}
\end{definition}

As in the case of Def.~\ref{def_calculus-of-relations}, one can adopt any appropriate set of inference rules for the equality to structure proofs in this calculus.

\begin{definition}[Relation algebras]
\label{def_relation-algebra}
The class of \emph{relation algebras} ($\RA$ for short) is the class of algebraic structures $\< A, \ajoin, \ameet, \acomple, \azero, \aunit, \acompo, \aconve, \aid\>$ satisfying the axioms in $\CR$.
\end{definition}

Tarski's first question about the relation between the class of models of the calculus of relations, $\RA$, and the class of concrete algebras of binary relations, $\PRA$, is of utmost importance in the context of computer science. Let us formulate it in more formal terms.

\begin{definition}
Given an algebra $\mathcal{A}$ and a class of algebras $\mathsf{K}$, $\mathcal{A}$ is \emph{representable} in $\mathsf{K}$ if there exists $\mathcal{B} \in \mathsf{K}$ such that $\mathcal{A}$ is isomorphic to $\mathcal{B}$. This notion generalises as follows: a class of algebras $\mathsf{K}_1$ is representable in a class of algebras $\mathsf{K}_2$ if every member of $\mathsf{K}_1$ is representable in $\mathsf{K}_2$. 
\end{definition}

Then, Tarski's first question explores whether the class $\RA$ is representable in the class $\PRA$. Lyndon's negative answer is devastating in practice. Consider the classical problem of formal verification in software engineering, stated as follows: given a specification of a software artefact written as a set of formulae $\Gamma \subseteq \mathit{CRForm}(\mathcal{R})$ (see Def.~\ref{eqcrform}) and a desired property of such an artefact, formalised as a formula $\alpha \in \mathit{CRForm}(\mathcal{R})$, does $\Gamma \vdash^\CR \alpha$? Then, we are interested in either constructing a proof for the previous judgement, or finding $\mathcal{A}$ such that $\mathcal{A} \models^\CR \Gamma$ and $\mathcal{A} \not\models^\CR \alpha$. In general, we would like $\mathcal{A}$ to be a concrete model (i.e. $\mathcal{A} \in \PRA$) as it provides a natural interpretation of relations in set theoretical terms. Lyndon's answer  
can be summarised as follows: It might happen that $\Gamma \not\vdash^\CR \alpha$ and for every $\mathcal{A} \in \RA$ such that $\mathcal{A} \models^\CR \Gamma$ and $\mathcal{A} \not\models^\CR \alpha$, there is no $\mathcal{B} \in \PRA$ such that $\mathcal{A} \simeq \mathcal{B}$ (i.e. every counterexample witnessing that $\Gamma \not\vdash^\CR \alpha$ is a non-representable relation algebra and thus, a model of no interest in this context).\\

The following definitions and properties will be of interest in further sections of the paper.

\begin{definition}[Full proper relation algebras]{\ \ \ \ \ \ \ \ \ } 
\label{fullPRA}
An algebraic structure $\< A, \pjoin, \pmeet, \pcomple, \pzero, \punit{E}, \pcompo, \pconve, \pid\>$ is said to be a \emph{full proper algebra of relations} over a set $S$ if:
\begin{itemize}
\item $A = \setof{a}{a \subseteq S \times S}$ (equivalently $A = \wp \(S^2\)$), and
\item $\punit{E} = \setof{\<a, b\>}{a, b \in S}$ (equivalently $E = S \times S$ or $E = S^2$).
\end{itemize}
\end{definition}

\begin{definition}[Ideal elements]
Let $\mathcal{A} \in \RA$ with domain $A$, then $x \in A$ is an \emph{ideal element} if and only if $x = \aunit \acompo x \acompo \aunit$.
\end{definition}

A formal definition of \emph{simple algebra} can be found in \cite[Chap.~2, Sec.~8]{burris81}, but for all practical purposes we will use the following theorem proved by Tarski.

\begin{theorem}[\cite{jonsson:ajm-74}, Thm.~4.10]
Let $\mathcal{A} \in \RA$, then $\mathcal{A}$ is \emph{simple} if and only if $\mathcal{A}$ has, at most, two ideal elements (i.e. $\azero$ and $\aunit$).
\end{theorem}

Notice that if a proper relation algebra is full, then it is simple. The proof of this property can be found in \cite[Thm.~4.10]{jonsson:ajm-74}. 

\begin{definition}
$\mathbf{1}$ and $\mathbf{2}$ are the only relation algebras with $1$ and $2$ elements in their domain, respectively.
\end{definition}

To ease the reader's understanding, when we interpret $\mathbf{1}$ as a proper relation algebra we obtain $\<\{\emptyset\}, \cup, \cap, \pcomple, \emptyset, \emptyset, \pcompo, \pconve, \emptyset\>$, with $\emptyset$ being the empty relation. On the other hand, if $U = \{\bullet\}$, $\mathbf{2}$ is interpreted as a proper relation algebra over $U$ as $\<\{\emptyset, \{\<\bullet, \bullet\>\}\}, \cup, \cap, \pcomple, \emptyset, \{\<\bullet, \bullet\>\}, \pcompo, \pconve, \{\<\bullet, \bullet\>\}\>$, also with $\emptyset$ being the empty relation.

\begin{definition}[Trivial and prime relation algebras]
\label{prime-alg}
Let $\mathcal{A} \in \RA$, then 
\begin{itemize}
\item $\mathcal{A}$ is \emph{trivial} if it is isomorphic to $\mathbf{1}$ (noted as $\mathcal{A} \simeq \mathbf{1}$) or to $\mathbf{2}$, and
\item $\mathcal{A}$ is \emph{prime} if it is simple and non-trivial.
\end{itemize}
\end{definition}

\vspace{0.5cm}
In \cite[Sec.~3]{frias02}, Frias pinpoints the motivations behind the introduction of \emph{Fork algebra}. In \cite{haeberer:ifiptc291}, Haeberer and Veloso started the study of this class of algebras in the search for a calculus suitable for program construction, derivation and verification.

If we recall Def.~\ref{def_proper-relation-algebra}, a \emph{proper fork algebra} is obtained by extending a proper relation algebra with a new operation called \emph{fork} and usually symbolised with ``$\afork$''. The introduction of this new operator induces a structure on the set over which the relations are defined. This is done by considering binary relations over the domain of a structure $\<U, \star\>$ where $\star: U \times U \to U$ is injective. Then, given $r, s \in \wp (U \times U)$,
\begin{equation}
\label{def_pfork}
r \pfork s = \setof{\pair{a}{b} \in U \times U}{\exists x, y \in U\ |\ b = x \star y \land a \; r \; x \land a \; s \; y}
\end{equation}

Fork algebras evolved around the definition of the function $\star$. In \cite{haeberer:ifiptc291} proper fork algebras were presented on a domain of binary relations on the set of finite trees built up from applications of $\star$; in that sense $\star$ acts as a set theoretical pairing function. In \cite{veloso:bsl-20_2} Veloso and Haeberer moved to a definition where the domain is built from binary relations on finite strings; an immediate consequence of this decision is that $\star$ acts as string concatenation. Later on, in \cite{veloso:19-92} the base set is once again made from finite trees. In all the previously mentioned articles, no axiomatization is given. Mikul\'{a}s et. al. proved, in \cite[Thm.~3.4]{mikulas:s+-192_7}, that an extension of a proper relation algebra with projection operators, like the ones induced by the operator $\pfork$ as defined in Eq.~\ref{def_pfork}, is not finitely axiomatizable. Such a result necessarily excludes the previous, more intuitive, interpretations of $\star$ as being binary tree constructor, string concatenation, set-theoretical pair formation, etc. 

If $U$ is the base set of a fork algebra and $x \in U$, $x$ is said to be a \emph{urelement} if there are no $y, z \in A$ such that $x = y \star z$. Intuitively, a urelement is a non-splitting element of $A$. It is easy to prove that having urelements is equivalent to having a non-surjective function $\star$. 

Before introducing proper fork algebras we introduce \emph{star proper fork algebras} as follows.

\begin{definition}[Star proper fork algebra]
\label{sPFA}
A \emph{star proper fork algebra} is a two-sorted algebraic structure $\< A, U, \pjoin, \pmeet, \pcomple, \pzero, \punit{E}, \pcompo, \pconve, \pid, \pfork, \star \>$ in which $A$ is a set of binary relations on $U$; $\pjoin$, $\pmeet$, $\pcompo$ and $\pfork$ are binary operations on $A$; $\pcomple$ and $\pconve$ are unary operations on $A$; $\pzero$, $\punit{E}$ and $\pid$ are distinguished elements of $A$; and $\star$ is a binary operation on $U$ satisfying:
\begin{itemize}
\item $\< A, \pjoin, \pmeet, \pcomple, \pzero, \punit{E}, \pcompo, \pconve, \pid \>$ is a proper relation algebra on $U$,
\item $\star: U \times U \to U$ is injective on the restriction of its domain to $\punit{E}$ and
\item $A$ is closed under $\pfork$ of binary relations, defined as follows:
$$r \pfork s = \setof{\pair{a}{x \star y} \in U \times U}{a \; r \; x \land a \; s \; y}\ .$$ 
\end{itemize}
If in addition, if $\star$ is required to be non-surjective, the algebra is referred to as a \emph{star proper fork algebra with urelements}. The class of star proper fork algebras (resp. star proper fork algebras with urelements) will be denoted as $\sPFA$ (resp. $\sPFAU$).
\end{definition}

A graphical interpretation of the fork of binary relations is presented in Fig.~\ref{fig_fork}.

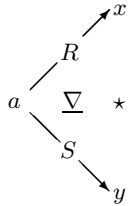
\begin{figure}[ht]
\unitlength1cm
\begin{center}
\begin{picture}(3.8,3)(0.1,-0.25)
\put(0.7,1.15){$a$}
\put(1,1.4){\line(1,1){0.36}}
\put(1.7,2.11){\vector(1,1){0.36}}
\put(1,1.1){\line(1,-1){0.4}}
\put(1.65,0.45){\vector(1,-1){0.4}}
\put(2.1,2.4){$x$}
\put(2.1,-0.05){$y$}
\put(1.4,1.15){$\pfork$}
\put(2.1,1.15){$\star$}
\put(1.4,1.8){$R$}
\put(1.4,0.5){$S$}
\end{picture}
\end{center}
\caption{Graphical representation of ``fork''.}
\label{fig_fork}
\end{figure}

When dealing with algebras, the function $\mathbf{Rd}_T$ takes reducts to the type $T$; to take reduct of a class of algebras of type $\< \mathcal{A}, \mathcal{F} \>$, to the type $\< \mathcal{A}', \mathcal{F'} \>$\footnote{For a formal definition of the type of an algebraic language see \cite[Defs.~1.1--1.3]{burris81}.} means to forget all the domains of $\mathcal{A}$ not mentioned in $\mathcal{A'}$ and the functions of $\mathcal{F}$ not mentioned in $\mathcal{F'}$. Notice that this definition requires $\mathcal{A'}$ to be included in $\mathcal{A}$, and $\mathcal{F'}$ to be included in $\mathcal{F}$.

\begin{definition}[Proper fork algebras]
\label{def_proper-fork-algebras}
The class of \emph{proper fork algebras} (denoted as $\PFA$) is obtained from $\sPFA$ as $\mathbf{Rd}_T\ \sPFA$, where $T$ is the similarity type $\< A, \pjoin, \pmeet, \pcomple, \pzero, \punit{E}, \pcompo, \pconve, \pid, \pfork \>$ and $\mathbf{Rd}_T$ is the algebraic operator for taking the reduct of an algebraic structure to the similarity type $T$.

Analogously, the class of \emph{proper fork algebras with urelements} (denoted as $\PFAU$) is obtained in the same way but applying the operator $\mathbf{Rd}_T$ on the class $\sPFAU$.
\end{definition}

\begin{definition}[Full proper fork algebras]\ \\
A proper fork algebra $\< A, \pjoin, \pmeet, \pcomple, \pzero, \punit{E}, \pcompo, \pconve, \pid, \pfork \>$ is said to be \emph{full} if its relational reduct (which is a proper relation algebra) $\< A, \pjoin, \pmeet, \pcomple, \pzero, \punit{E}, \pcompo, \pconve, \pid \>$ is full.
\end{definition}

The class of \emph{full proper fork algebras} (resp. \emph{full proper fork algebras with urelements}) will be denoted as $\fPFA$ ($\fPFAU$).\\

In the same way $\RA$ is the class of models of $\CR$, and $\PRA$ its class of concrete models, one can formalise a calculus of fork algebras (resp. calculus of fork algebras with urelements) and establish the formal relationship between its class of abstract models and $\PFA$. As Frias points out in \cite{frias02}, the current, most accepted, axiomatisation for the fork algebras is the one due to Haeberer et al. \cite{haeberer:fmpta93,haeberer:4-93}.

\begin{definition}
Let $\mathcal{R}$ be a set of relation variables, then the set of relation designations is the smallest set $\mathit{RelDes}(\mathcal{R})$
such that:
\begin{itemize}
\item $\mathcal{R} \cup \set{\aunit, \azero, \aid} \subseteq \mathit{RelDes}(\mathcal{R})$,
\item If $r, s \in \mathit{RelDes}(\mathcal{R})$, then $\set{r \ajoin s, r \ameet s, \acompl{r}, r \acompo s, \aconv{r}, r \afork s} \subseteq \mathit{RelDes}(\mathcal{R})$.
\end{itemize}
Then, the set of formulae is the set $\setof{r = s}{r, s \in \mathit{CFAForm}(\mathcal{R})}$.
\end{definition}

\begin{definition}[The calculus of fork algebras]
\label{def_calculus-of-fork-algebras}
Let $\mathcal{R}$ be a set of relation variables, then the \emph{calculus of fork algebras} ($\CFA$ for short) is defined for the formulae in $\mathit{CFAForm}(\mathcal{R})$ by:
\begin{itemize}
\item the axioms for the Boolean and the relational operators of Def.~\ref{def_new-calculus-of-relations}, and
\item the following axioms for the fork operator: for all $r, s, t, u \in A$
$$
\begin{array}{l}
r \afork s = (r \acompo (\aid \afork \aunit)) \ameet (s \acompo (\aunit \afork \aid))\\
(r \afork s) \acompo \aconv{t \afork u} = (r \acompo \aconv{t}) \ameet (s \acompo \aconv{u})\\
\aconv{\aid \afork \aunit} \afork \aconv{\aunit \afork \aid} \leq \aid
\end{array}
$$
\end{itemize}
Additionally, the \emph{calculus of fork algebras with urelements} ($\CFAU$ for short) is obtained by adding the axiom:
$$\aunit \acompo (\acompl{\aunit \afork \aunit} \ameet \aid) \acompo \aunit = \aunit$$
\end{definition}

Once again, the calculus is completed by adopting any appropriate set of inference rules for the equality.

\begin{definition}
\label{def_fork-algebras}
The class of \emph{fork algebras}, denoted as $\FA$ for short, (resp. \emph{fork algebras with urelements}, denoted as $\FAU$) is the class of algebraic structures $\< A, \ajoin, \ameet, \acomple, \azero, \aunit, \acompo, \aconve, \aid, \afork\>$ satisfying the axioms in $\CFA$ (resp. $\CFAU$).
\end{definition}

The term $\acompl{\aunit \afork \aunit} \ameet \aid$, appearing in the last axiom of Def.~\ref{def_calculus-of-fork-algebras}, characterises the partial identity on urelements. This term will be denoted by $\aidu$. Terms $\aconv{\aid \afork \aunit}$ and $\aconv{\aunit \afork \aid}$, when interpreted in a proper fork algebra, act as projections of the first and second coordinates, respectively, of an element obtained by application of $\star$. These two terms will be denoted by $\pi$ and $\rho$, respectively. Figs.~\ref{fig_pi}  and~\ref{fig_rho} show a graphical representation of projections $\pi$ and $\rho$.

\begin{figure}[ht]
   \centering
   \subfloat[Graphical representation of $\pi$.\label{fig_pi}]{
     \begin{minipage}{0.4\textwidth}
     \centering
\unitlength1cm
\begin{picture}(2.0,3)(0,0)
\put(0.5,2.8){\line(1,-1){0.4}}
\put(1.18,2.1){\vector(1,-1){0.5}}
\put(0.5,0.2){\line(1,1){0.36}}
\put(1.22,0.89){\vector(1,1){0.46}}
\put(0.25,2.75){$x$} \put(0.25,1.48){$\star$}
\put(0.25,0.15){$y$} \put(0.87,1.45){$\pfork$}
\put(1.75,1.45){$x$} \put(0.9,2.2){$\pid$}
\put(0.9,0.57){$\punit{E}$}
\end{picture}
\end{minipage}
}
\hspace{1cm}
   \subfloat[Graphical representation of $\rho$.\label{fig_rho}]{
     \begin{minipage}{0.4\textwidth}
     \centering
\unitlength1cm
\begin{picture}(2.0,3)(0,0)
\put(0.5,2.8){\line(1,-1){0.4}}
\put(1.19,2.12){\vector(1,-1){0.5}}
\put(0.5,0.2){\line(1,1){0.36}}
\put(1.3,0.97){\vector(1,1){0.38}}
\put(0.25,2.75){$x$} \put(0.25,1.48){$\star$}
\put(0.25,0.15){$y$} \put(0.87,1.45){$\pfork$}
\put(1.75,1.45){$y$} \put(0.92,2.15){$\punit{E}$}
\put(0.87,0.62){$\pid$}
\end{picture}
\end{minipage}
}
\caption{Graphical representation of projections.}\label{fig:projections}
\end{figure}
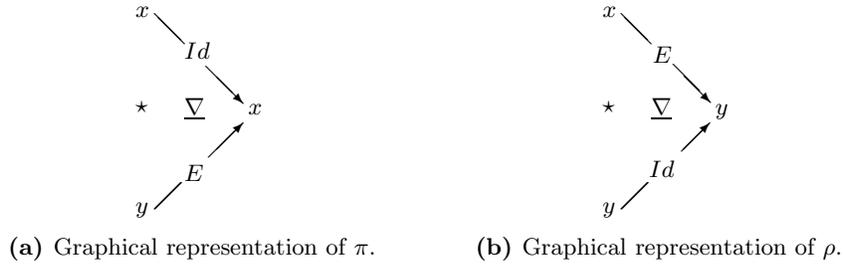

These definitions allow us to rewrite the first, third and fourth axiom of Def~\ref{def_calculus-of-fork-algebras} as follows:
$$
\begin{array}{l}
r \afork s = (r \acompo \aconv{\pi}) \ameet (s \acompo \aconv{\rho})\\
\pi \afork \rho \leq \aid\\
\aunit \acompo \aidu \acompo \aunit = \aunit
\end{array}
$$

By resorting to the identity between urelements we define $\uaunitu = \aidu \acompo \aunit \acompo \aidu$. Relation $\uaunitu$ relates every pair of urelements.

Checking that proper fork algebras (resp. proper fork algebras with urelements) are fork algebras (resp. fork algebras with urelements) is simple as it only requires to check that the structures defined in Def.~\ref{def_proper-fork-algebras} satisfy the axioms given in Def.~\ref{def_calculus-of-fork-algebras}. In \cite{frias:bsl-24_2}, Frias et al. proved the converse result by showing that $\FA$ is representable in $\PFA$, but resorting to a non-equational axiom. Later on, in \cite{frias:bsl-24_4}, the same representability result was proved but only resorting to those equational axioms appearing in \cite{haeberer:fmpta93}. 

\begin{theorem}[Representability of $\FA$ in $\PFA$, \cite{frias:bsl-24_4}, Thm.~3.7\footnote{The same representability result was obtained independently by Gyuris and presented in \cite{gyuris:tcs-188_1_2}.}]\ \\
\label{representability}
$\FA = \mathbf{I}\ \PFA$\footnote{The operator $\mathbf{I}$ closes a class of algebras of the same similarity type under isomorphisms.}
\qed
\end{theorem}

The proof of Thm.~\ref{representability}, published also in \cite[Sec.~4.1]{frias02}, can be easily adapted to a proof of the representability of $\FAU$ in $\PFAU$.

\begin{corollary}\ \\
\label{representability-urelements}
$\FAU = \mathbf{I}\ \PFAU$.
\qed
\end{corollary}

\vspace{0.5cm}
The relational reduct of a fork algebra satisfies many specific properties, for example, they are representable in $\PRA$. In \cite{veloso:ES-418-96}, and briefly revisited in \cite{veloso:bcml96}, Veloso explores the expansibility of a relation algebra to a fork algebra and, to that purpose, he presents the following definitions and results.

\begin{definition}
Given $\mathcal{F} \in \FA$, $\mathbf{Rd}_{\< A, \ajoin, \ameet, \acomple, \azero, \aunit, \acompo, \aconve, \aid \>}\ \mathcal{F}$ (called its relational reduct) will be denoted as $\mathcal{F}_\mathsf{RA}$.
\end{definition}

\begin{definition}[Fork index]
Let $\mathcal{A} \in \RA$, the \emph{fork index} $\varphi$ of $\mathcal{A}$ is defined as $\varphi (\mathcal{A}) = |\setof{\mathcal{F} \in \FA}{\mathcal{F}_\mathsf{RA} = \mathcal{A}}|$. 
\end{definition}

\begin{proposition}[\cite{veloso:ES-418-96}, Sec.~5.1]
Let $\mathcal{A} \in \RA$ with domain $A$, Then $\varphi (\mathcal{A}) \leq |A|^2$.
\qed
\end{proposition}

With this result, Veloso introduces the following classification of relation algebras.

\begin{definition}
Let $\mathcal{A} \in \RA$ with domain $A$,
\begin{itemize}
\item $\mathcal{A}$ is called \emph{explosive} if $\varphi (\mathcal{A}) = |A|^2$,
\item $\mathcal{A}$ is called \emph{non-expansible} if $\varphi (\mathcal{A}) = 0$,
\item $\mathcal{A}$ is called \emph{rigid} if $\varphi (\mathcal{A}) = 1$,
\item $\mathcal{A}$ is called \emph{elastic} if $\varphi (\mathcal{A}) = \infty$.
\end{itemize}
\end{definition}

The aim of the present work is to contribute to the study of the first category of relation algebras defined above, i.e. the explosive relation algebras. Following \cite{veloso:05-96,veloso:ES-418-96}, the class of explosive relation algebras will be denoted as $\mathsf{EXP}$ and, given $\kappa \geq \aleph_0$, $\mathsf{EXP}[\kappa]$ will denote the class of explosive relation algebras whose underlying set is of cardinality $\kappa$.


\section{On the construction of explosive relation algebras}
\label{EXPconstruction}
In \cite[Sec.~6.2]{veloso:ES-418-96} Veloso presents the following three results. The first one, to be discussed more extensively in Sec.~\ref{sec:fixpoint-star}, constitutes the main motivation behind this work.

\begin{proposition}[Existence of prime, big and explosive proper relation algebras, \cite{veloso:ES-418-96}, Sec.~6.2]
\label{existsEXP}
Let $\kappa \geq \aleph_0$, there exists $\mathcal{R}_\kappa \in \PRA$ prime and explosive such that $|\mathcal{R}_\kappa| = \kappa$ (i.e. $\mathcal{R}_\kappa$ has $\kappa$ non-isomorphic expansions to fork algebras $\{\mathcal{F}_\gamma\}_{\gamma < \kappa}$).\qed
\end{proposition}

The next property is a direct consequence of the fact that proper relation algebras having a different number of ideals cannot be isomorphic (the interested reader is pointed to \cite{jonsson:ajm-74} for a discussion about the relation between homomorphisms and ideal elements). Then, to control the amount of ideal elements, we can combine prime algebras (see Prop.~\ref{prime-alg}) with powers of $\mathbf{2}$ in a direct product.

\begin{proposition}[Non-isomorphic combinations of prime, big and explosive proper relation algebras, \cite{veloso:ES-418-96}, Sec.~6.2]
\label{existsEXP2}
Let $\kappa \geq \aleph_0$ and $\mathcal{R}_\kappa \in \PRA$ prime and explosive such that $|\mathcal{R}_\kappa| = \kappa$; then for each cardinal $\zeta < \kappa$, $\mathbf{2}^\zeta \times \mathcal{R}_\kappa \in \PRA$ is representable, explosive, $|\mathbf{2}^\zeta \times \mathcal{R}_\kappa| = \kappa$ and has $2^{\zeta+1}$ ideal elements.
\qed
\end{proposition}

We will not discuss why direct product (used in the previous proposition) does not modifies the amount of posible expansions of a proper relation algebra to a fork algebra, the reader interested in the details of such phenomenon is pointed to \cite{veloso:05-96,veloso:ES-418-96}.

\begin{theorem}[Many prime, big  and explosive proper relation algebras, \cite{veloso:ES-418-96}, Sec.~6.2]
\label{existsEXP3}
Let $\kappa \geq \aleph_0$, then there exists $\kappa$ non-isomorphic proper relation algebras of cardinality $\kappa$ (i.e $|\mathsf{EXP}[\kappa]| = \kappa$)
\qed
\end{theorem}

\subsection{A prime, big and explosive proper relation algebra}
\label{sec:fixpoint-star}
In this section we review the construction of a prime, big and explosive proper relation algebra, presented by Veloso in \cite{veloso:05-96} and used in the proof of Prop.~\ref{existsEXP}.

\begin{definition}[$\underline{2}$]
Let $\mathcal{F} = \< A, \pjoin, \pmeet, \pcomple, \pzero, \punit{E}, \pcompo, \pconve, \pid, \pfork \> \in \PFA$, then $\underline{2} = \pid \pfork \pid$. 
\end{definition}

\begin{definition}[Subidentities of $\underline{2}$]
Let $\mathcal{F} = \< A, \pjoin, \pmeet, \pcomple, \pzero, \punit{E}, \pcompo, \pconve, \pid, \pfork \> \in \PFA$, then $\mathit{Si}_{\underline{2}} (\mathcal{F}) = \setof{a \in A}{a \subseteq \underline{2} \pmeet \pid}$.
\end{definition}

\begin{proposition}
Let $\mathcal{F}, \mathcal{G} \in \PFA$, if $\phi: \mathcal{F} \to \mathcal{G}$ is an isomorphism, then $\phi$ induces a bijection between $\mathit{Si}_{\underline{2}} (\mathcal{F})$ and $\mathit{Si}_{\underline{2}} (\mathcal{G})$.
\qed
\end{proposition}

\begin{definition}[Fixpoints of $\star$]
\label{fixpoints}
Let $\star: U^2 \to U$, the fixpoints of $\star$ are defined as $\mathit{fix} (\star) = \setof{u \in U}{u \star u = u}$.
\end{definition}

This set can also be presented as a relation contained in the identity relation, as follows:
\begin{equation}
\label{Id-fixpoints}
\pid_{\mathit{fix}(\star)} = \setof{\langle u, u \rangle \in U^2}{u \in \mathit{fix}(\star)}
\end{equation}
\noindent for which it is possible to prove the following properties.

\begin{proposition}
Let $\mathcal{F} = \< A, \pjoin, \pmeet, \pcomple, \pzero, \punit{E}, \pcompo, \pconve, \pid, \pfork \> \in \PFA$, then $\underline{2} \pmeet \pid = \pid_{\mathit{fix} (\star)}$.
\qed
\end{proposition}

\begin{proposition}
Let $\mathcal{F} = \< A, \pjoin, \pmeet, \pcomple, \pzero, \punit{E}, \pcompo, \pconve, \pid, \pfork \> \in \PFA$ simple, such that $\pfork$ is induced by $\star: U^2 \to U$, then $\mathit{Si}_{\underline{2}} (\mathcal{F}) = \wp (\pid_{\mathit{fix} (\star)}) \pmeet A$.
\qed
\end{proposition}

\begin{proposition}[\cite{veloso:05-96}, Sec.~6]
\label{existsEXPconst}
Let $U$ be a set such that $|U| = \kappa$ and $\aleph_0 \leq \kappa$ then, for all $S \subseteq U$ such that $|S| < |U|$, there exists $\star_S: U^2 \to U$ bijective such that $\mathit{fix}(\star_S) = S$.
\end{proposition}
\begin{proof}
If $|U| = \kappa$, then $|U^2| = \kappa$, $|\pid| = \kappa$ and $|\overline{Id}| = \kappa$. Then, let $S \subseteq U$ such that $|S| < \kappa$, we know that the cardinality of the complement of $S$ with respect to $U$ (denoted as $\overline{S}$ when no ambiguity arises) is $\kappa$ (denoted as $|\overline{S}| = \kappa$) and, therefore, it is possible to take $\overline{S} = A \cup \bigcup_{i \in \nat} B_i$ such that:
\begin{itemize}
\item for all $i \in \nat$, $A \cap B_i = \emptyset$,
\item for all $i, j \in \nat$, such that $i \not= j$, $B_i \cap B_j = \emptyset$, 
\item $|A| = \kappa$ and for all $i \in \nat$, $|B_i| = \kappa$.
\end{itemize}
Then, there exists $g: \overline{S} \to \bigcup_{i \in \nat} B_i$ bijective and without fixpoints defined as the union of the bijective functions $g_A: A \to B_0$ and $\{g_i: B_i \to B_{i+1}\}_{i \in \nat}$. Notice that such a union is disjoint as the domains and codomains of all of the functions are disjoint.

On the other hand, there exists $f: \pcompl{\pid} \to A$, as a consequence of analysing the cardinality of $\pcompl{\pid}$ and $A$ and, consequently, it is possible to define $\star_S: U^2 \to U$ in the following way:
\[
\begin{array}{rcl}
u \star_S v & = & \left\{
\begin{array}{lr}
u & \mbox{, if $u \in S$ and $v = u$.}\\
g (u) & \mbox{, if $u \in \overline{S}$ and $v = u$.}\\
f (\langle u, v \rangle) & \mbox{, if $v \not= u$.}
\end{array}
\right.
\end{array}
\]
From this we know that $\star_S:U^2 \to U$ is a bijective function as it is the union of bijective functions whose domains and codomains are pairwise disjoint such that $\mathit{fix}(\star_S) = S$.
\qed\end{proof}

Now, it is possible to prove Prop.~\ref{existsEXP} by resorting to the construction of a proper relation algebra with the corresponding properties.

\thmt{Proposition}{existsEXP}{Existence of prime, big and explosive proper relation algebras, \cite{veloso:ES-418-96}, Sec.~6.2}{Let $\kappa \geq \aleph_0$, there exists $\mathcal{R}_\kappa \in \PRA$ prime and explosive such that $|\mathcal{R}_\kappa| = \kappa$ (i.e. $\mathcal{R}_\kappa$ has $\kappa$ non-isomorphic expansions to fork algebras $\{\mathcal{F}_\gamma\}_{\gamma < \kappa}$).}
\begin{proof}
Let $U$ be an infinite set such that $|U| = \kappa$ then, for all $\phi < \kappa$, there exists $S \subseteq U$ such that $|S| = \phi$. By Prop.~\ref{existsEXPconst}, there exists $\star_S:U^2 \to U$ bijective inducing $\pfork_S: \(\wp (U^2)\)^2 \to \wp (U^2)$ and the corresponding projections $\pi_S$ and $\rho_S$.

Let $\{S_\phi\}_{\aleph_0 \leq \phi \leq \kappa}$ such that $S_\phi \subseteq U$, for all $\aleph_0 \leq \phi \leq \kappa$ then, we define the set $H$ in the following way\footnote{$\wp_\mathit{fin} (A)$ is interpreted as $\setof{a}{a \subseteq A \land |a| \in \nat}$, the finite powerset of $A$.}:
\[
H = \wp_\mathit{fin} \(U^2\) \cup \bigcup_{\phi < \kappa} \(\{\pi_{S_\phi}, \rho_{S_\phi}\} \cup \wp \(\pid_{\mathit{fix} (\star_{S_\phi})}\)\)
\]
From the previous definition we know that $|H| = \kappa$. Therefore, it is enough to consider $\mathcal{R}_H$ as the subalgebra generated by $H$ of the full proper relation algebra generated by $U$ which, by \cite[Sec.~3]{burris81} has cardinality $\kappa$. Finally, by \cite[Thm.~4.11]{jonsson:ajm-74}, as $\mathcal{R}_H$ is a subalgebra of a simple algebra, it is simple, and for each $\aleph_0 \leq \phi \leq \kappa$, $S_\phi$ induces $\pfork_{S_\phi}$ determining a possible extension of $\mathcal{R}_H$ to a fork algebra. Notice that all this possible extensions are pairwise non-isomorphic as they differ in the cardinality of the set of fixpoints for $\star$.
\qed\end{proof}

Next, we reproduce the proof of Prop.~\ref{existsEXP2} and Thm.~\ref{existsEXP3}.

\thmt{Proposition}{existsEXP2}{Non-isomorphic combinations of prime, big and explosive proper relation algebras, \cite{veloso:ES-418-96}, Sec.~6.2}{
Let $\kappa \geq \aleph_0$ and $\mathcal{R}_\kappa \in \PRA$ prime and explosive such that $|\mathcal{R}_\kappa| = \kappa$; then for all cardinal $\zeta < \kappa$, $\mathbf{2}^\zeta \times \mathcal{R}_\kappa \in \PRA$ is representable, explosive, $|\mathbf{2}^\zeta \times \mathcal{R}_\kappa| = \kappa$ and has $2^{\zeta+1}$ ideal elements.
}
\begin{proof}
Since $\mathbf{2}^\zeta \leq \kappa$, the direct product $\mathbf{2}^\zeta \times \mathcal{R}_\kappa$ is a representable relation algebra such that $|\mathbf{2}^\zeta \times \mathcal{R}_\kappa| = \kappa$; it's prime factors are the rigid $\mathbf{2}$ and the explosive $\mathcal{R}_\kappa$, so $\phi (\mathbf{2}^\zeta \times \mathcal{R}_\kappa) = \kappa$ leading to the fact that $\mathbf{2}^\zeta \times \mathcal{R}_\kappa$ has $2^\zeta \cdot 2$ ideal elements.
\qed\end{proof}

\thmt{Theorem}{existsEXP3}{Many prime, big  and explosive proper relation algebras, \cite{veloso:ES-418-96}, Sec.~6.2}{
Let $\kappa \geq \aleph_0$, then there exists $\kappa$ non-isomorphic proper relation algebras of cardinality $\kappa$ (i.e $|\mathsf{EXP}[\kappa]| = \kappa$)
}
\begin{proof}
This theorem is a direct consequence of Prop.~\ref{existsEXP2}, as algebras with different cardinality of ideal elements cannot be isomorphic.
\qed\end{proof}

\subsection{Generalising the control of the fixpoints of $\star$}
In the previous section we have shown the construction of a prime, big and explosive proper relation algebra where explosiveness is guaranteed by controlling the cardinality of the set of fixpoints of $\star$, each of which leads to a non-isomorphic fork algebra. In this section we propose a generalisation of the controlling technique of such fixpoints.

Let us first introduce some useful definitions.

\newcommand\bt{\mathsf{bt}}
\newcommand\BT{\mathsf{BT}}
\newcommand\nil{\mathtt{nil}}
\newcommand\bin[2]{\mathtt{bin}\ #1\ #2}

\begin{definition}[Binary trees]
\emph{Binary trees} are the elements of $\BT$, the smallest set of terms produced by the grammar $\bt ::= \nil\ |\ \bin{\bt}{\bt}$. 
\end{definition}

\begin{definition}
The predicates $\bullet = \bullet \subseteq \BT \times \BT$ and $\bullet < \bullet \subseteq \BT \times \BT$ are defined as follows:
\[
\begin{array}{rclcl}
\nil & = & \nil\\
\bin{i}{d} & = & \bin{i'}{d'} & \text{ iff } & i = i' \text{ and } d = d'\\

\nil & < & \bin{i}{d}\\
\bin{i}{d} & < & \bin{i'}{d'} & \text{ iff } & \bin{i}{d} = i' \text{ or } \bin{i}{d} < i' \text{ or}\\
              &    &                  &                & \bin{i}{d} = d' \text{ or } \bin{i}{d} < d'
\end{array}
\]
\end{definition}

\newcommand\map[3]{\mathit{map}\ #1\ #2\ #3}

\begin{definition}[Map]
\label{map}
Let $U$ be a set, we define $\mathit{map}: \BT \times [U^2 \to U] \times U \to U$ as follows: let $f: U^2 \to U$ a binary function over $U$ and $u \in U$
\[
\begin{array}{rcl}
\map{\nil}{f}{u} & = & u\\
\map{\(\bin{\mathit{ab}_1}{\mathit{ab}_2}\)}{f}{u} & = & f \(\(\map{\mathit{ab}_1}{f}{u}\), \(\map{\mathit{ab}_2}{f}{u}\)\)
\end{array}
\]
\end{definition}

\begin{definition}[$\underline{t}$]
\label{term-underline}
Let $\mathcal{F} = \< A, \pjoin, \pmeet, \pcomple, \pzero, \punit{E}, \pcompo, \pconve, \pid, \pfork \> \in \PFA$ and $t \in \BT$, then $\underline{t} = \mathit{map} (t, \pfork, \pid)$.
\end{definition}

\begin{definition}[Subidentities of $\underline{t}$]
\label{def:sit}
Let $\mathcal{F} = \< A, \pjoin, \pmeet, \pcomple, \pzero, \punit{E}, \pcompo, \pconve, \pid, \pfork \> \in \PFA$ and $t \in \BT$, then $\mathit{Si}_{\underline{t}} (\mathcal{F}) = \setof{a \in A}{a \subseteq \underline{t} \pmeet \pid}$.
\end{definition}

\begin{proposition}
\label{isomorphism}
Let $\mathcal{F} = \< F, \pjoin^\mathcal{F}, \pmeet^\mathcal{F}, \pcomple^\mathcal{F}, \pzero^\mathcal{F}, \punit{E}^\mathcal{F}, \pcompo^\mathcal{F}, {\pconve}^\mathcal{F}, \pid^\mathcal{F}, \pfork^\mathcal{F} \>$ and $\mathcal{G} = \< G, \pjoin^\mathcal{G}, \pmeet^\mathcal{G}, \pcomple^\mathcal{G}, \pzero^\mathcal{G}, \punit{E}^\mathcal{G}, \pcompo^\mathcal{G}, {\pconve}^\mathcal{G}, \pid^\mathcal{G}, \pfork^\mathcal{G} \>$ be $\PFA$ and $t \in \BT$, if $\phi: \mathcal{F} \to \mathcal{G}$ is an isomorphism, then $\phi$ induces a bijection between $\mathit{Si}_{\underline{t}^\mathcal{F}} (F)$ and $\mathit{Si}_{\underline{t}^\mathcal{G}} (G)$.
\end{proposition}
\begin{proof}
Let $\phi: F \to G$ be an isomorphism between $\mathcal{F}$ and $\mathcal{G}$, and $r \in F$ such that $r \in \mathit{Si}_{\underline{t}^\mathcal{F}} (F)$,
\[
\begin{array}{rclr}
r \in \mathit{Si}_{\underline{t}^\mathcal{F}} (F) 
		& \text{ iff } & r \subseteq^\mathcal{F} \underline{t}^\mathcal{F} \pmeet^\mathcal{F} \pid^\mathcal{F} & \text{[by Def.~\ref{def:sit}.]}\\
		& \text{ iff } & r \pjoin^\mathcal{F} \(\underline{t}^\mathcal{F} \pmeet^\mathcal{F} \pid^\mathcal{F}\) = \underline{t}^\mathcal{F} \pmeet^\mathcal{F} \pid^\mathcal{F}  & \text{[by Def. of $\subseteq$.]}\\
		& \text{ iff } & \phi\(r \pjoin^\mathcal{F} \(\underline{t}^\mathcal{F} \pmeet^\mathcal{F} \pid^\mathcal{F}\)\) = \phi\(\underline{t}^\mathcal{F} \pmeet^\mathcal{F} \pid^\mathcal{F}\)  \\
		& &\qquad  \text{[because $\phi$ is an isomorphism.]}\\
		& \text{ iff } & \phi\(r\) \pjoin^\mathcal{G} \(\phi\(\underline{t}^\mathcal{F}\) \pmeet^\mathcal{G} \pid^\mathcal{G}\) = \phi\(\underline{t}^\mathcal{F}\) \pmeet^\mathcal{G} \pid^\mathcal{G}  \\
		& &\qquad  \text{[because $\phi$ is an isomorphism.]}\\
		& \text{ iff } & \phi\(r\) \pjoin^\mathcal{G} \(\underline{t}^\mathcal{G} \pmeet^\mathcal{G} \pid^\mathcal{G}\) = \underline{t}^\mathcal{G} \pmeet^\mathcal{G} \pid^\mathcal{G}  & \text{[by Lemma~\ref{bijection-t}.]}\\
		& \text{ iff } & \phi\(r\) \subseteq^\mathcal{G} \underline{t}^\mathcal{G} \pmeet^\mathcal{G} \pid^\mathcal{G}  & \text{[by Def. of $\subseteq$.]}\\
		& \text{ iff } & \phi\(r\) \in \mathit{Si}_{\underline{t}^\mathcal{G}} (G) & \text{[by Def.~\ref{def:sit}.]}
\end{array}
\]
Thus, finishing the proof.\qed\end{proof}

\begin{definition}[$t$-controlled fixpoints of $\star$]
\label{fixpoints-t}
Let $\star: U^2 \to U$ and $t \in \BT$, the $t$-controlled fixpoints of $\star$ are defined as $\mathit{fix}_t (\star) = \setof{u \in U}{\map{t}{\star}{u} = u}$.
\end{definition}

Recalling the definition of the set of fixpoints of $\star$ given in Def.~\ref{fixpoints} we can present the $t$-controlled fixpoints of $\star$ as a partial identity as follows.
\begin{equation}
\label{Id-fixpoints-t}
\pid_{\mathit{fix}_t (\star)} = \setof{\langle u, u \rangle \in U^2}{u \in \mathit{fix}_t (\star)}
\end{equation}
\noindent for which it is possible to derive the following properties.

\begin{proposition}
\label{prop-id}
Let $\mathcal{F} = \< F, \pjoin, \pmeet, \pcomple, \pzero, \punit{E}, \pcompo, \pconve, \pid, \pfork \> \in \PFA$ with $\pfork$ induced by $\star: U^2 \to U$ and $t \in \BT$ such that $t \not= \nil$, then $\underline{t} \pmeet \pid = \pid_{\mathit{fix}_t (\star)}$.
\end{proposition}
\begin{proof}
Let $\<u, v\> \in U^2$, then
\[
\begin{array}{rclr}
\<u, v\> \in \underline{t} \pmeet \pid & \text{ iff } & \< u, v \> \in \map{t}{\pfork}{\pid} \text{ and } u = v & \text{[by Defs.~\ref{term-underline} and~\ref{def_proper-relation-algebra}.]}\\
		& \text{ iff } & \< u, u \> \in \map{t}{\pfork}{\pid} \text{ and } u = v \\
		& \text{ iff } & u = \map{t}{\star}{u} \text{ and } u = v & \text{[by Lemma~\ref{fork2star}.]}\\
		& \text{ iff } & u \in \mathit{fix}_t (\star) \text{ and } u = v & \text{[by Def.~\ref{fixpoints-t}.]}\\
		& \text{ iff } & \langle u, u \rangle \in \pid_{\mathit{fix}_t (\star)} \text{ and } u = v &\text{[by Eq.~\ref{Id-fixpoints-t}.]}\\
		& \text{ iff } & \langle u, v \rangle \in \pid_{\mathit{fix}_t (\star)}\\
\end{array}
\]
Thus, finishing the proof.\qed\end{proof}

\begin{proposition}
\label{Sit-powerset}
Let $\mathcal{F} = \< F, \pjoin, \pmeet, \pcomple, \pzero, \punit{E}, \pcompo, \pconve, \pid, \pfork \> \in \PFA$ simple with $\pfork$ induced by $\star: U^2 \to U$ and $t \in \BT$ such that $t \not= \nil$, then $\mathit{Si}_{\underline{t}} (\mathcal{F}) = \wp \(\pid_{\mathit{fix}_t (\star)}\) \pmeet F$.
\end{proposition}
\begin{proof}
\[
\begin{array}{rclr}
\mathit{Si}_{\underline{t}} (\mathcal{F}) & = & \setof{r \in F}{r \subseteq \underline{t} \pmeet \pid} & \text{[by Defs.~\ref{def:sit}.]}\\
		& = & \setof{r \in F}{r \subseteq \pid_{\mathit{fix}_t (\star)}} & \text{[by Prop.~\ref{prop-id}.]} \\
		& = & \setof{r \in F}{r \in \wp\(\pid_{\mathit{fix}_t (\star)}\)} & \text{[by Def.~$\subseteq$.]} \\
		& = & \setof{r \in \wp \(U^2\)}{r \in \wp\(\pid_{\mathit{fix}_t (\star)} \text{ and } r \in F\)} \\
		& = & \setof{r \in \wp \(U^2\)}{r \in \wp\(\pid_{\mathit{fix}_t (\star)} \pmeet F\)} \\
		& = & \wp\(\pid_{\mathit{fix}_t (\star)} \pmeet F\)
\end{array}
\]
Thus, finishing the proof.\qed\end{proof}

The following property is analogous to Prop.~\ref{existsEXPconst} but relaxes the conditions over $\star$.

\begin{proposition}
\label{existsEXPtconst}
Let $U$ be an infinite set $|U| = \kappa$ and $\aleph_0 \leq \kappa$ then for all $S \subseteq U$ such that $|S| < |U|$, there exists $\star_S: U^2 \to U$ injective such that given $t \in \BT$ where $t \not= \nil$, $\left|\mathit{fix}_t \(\star_S\)\right| = |S|$.
\end{proposition}
\begin{proof}
If $|U| = \kappa$, then $|U^2| = \kappa$, $|\pid| = \kappa$ and $|\overline{Id}| = \kappa$. Then, let $\{S_{t'}\}_{t' < t}$ be a finite family of sets such that:
\begin{itemize}
\item for all $t', t'' < t$, if $t' \not= t''$ then $S_{t'} \cap S_{t''} = \emptyset$, and
\item for all $t' < t$, $|S_{t'}| = |S|$.
\end{itemize}
Analogous to what Veloso points out in the proof of Prop.~\ref{existsEXPconst}, we know that $|\overline{S \cup \bigcup_{t' < t} S_{t'}}| = \kappa$ and, therefore, it is possible to take $\overline{S \cup \bigcup_{t' < t} S_{t'}} = \bigcup_{i \in \nat} B_i$ such that:
\begin{itemize}
\item for all $i, j \in \nat$, such that $i \not= j$, $B_i \cap B_j = \emptyset$, and
\item $|A| = \kappa$ and for all $i \in \nat$, $|B_i| = \kappa$.
\end{itemize}
Then, there exists $g: \bigcup_{i \in \nat} B_i \to \bigcup_{i \in \nat} B_i$ bijective and without fixpoints defined as the union of the bijective functions $\{g_i: B_i \to B_{i+1}\}_{i \in \nat}$. Notice that such a union is disjoint as the domains and codomains of all of the functions are disjoint.

On the other hand, there exists an injective funtion $f: S^2 \to B_0$ and a finite family of bijective functions $\{h_{t'}: S \to S_{t'}\}_{t' < t}$, such that $h_\nil = id_S$. Then, it is possible to define $\star_S: U^2 \to U$ according to Table~\ref{table1}.
\begin{table}[ht]
\centering
\begin{tabular}{|c||c|c|}
\hline
$u \star_S v$ & $S_{t''}$ & $B_i$ \\
\hline
\hline
$S_{t'}$ & $\left\{ \begin{array}{l} 
                         h_{\bin{t'}{t''}} \({h_{t'}}^{-1} (u)\) \\ \qquad\mbox{; if ${h_{t'}}^{-1} (u) = {h_{t''}}^{-1} (v)$ and}\\ \qquad\quad \mbox{$\bin{t'}{t''} < t$.} \\
                         {h_{t'}}^{-1} (u)                            \\ \qquad\mbox{; if ${h_{t'}}^{-1} (u) = {h_{t''}}^{-1} (v)$ and}\\ \qquad\quad \mbox{$\bin{t'}{t''} = t$.} \\
                         f ({h_{t'}}^{-1} (u), {h_{t''}}^-1 (v)) \\ \qquad\mbox{; otherwise.}
                   \end{array} \right.$ & $g(v)$ \\
\hline
$B_j$    & $g (u)$ & $\left\{ \begin{array}{ll}
                                      g (u) & \mbox{; if $i \geq j$.}\\
                                      g (v) & \mbox{; otherwise.}
                                   \end{array}\right.$ \\
\hline
\end{tabular}
\caption{Definition of $\star_S: U^2 \to U$ controlled by a binary tree.}
\label{table1}
\end{table}

Table~\ref{table1} shows how $\star_S: U^2 \to U$ is defined as the union of injective functions with disjoint domains and codomains. Therefore, $\star_S: U^2 \to U$ is injective.

On the one hand, by Lemma~\ref{Sfixpoint} we obtain that if $s \in S$, $s = \map{t}{\star_S}{s}$ and, by Def.~\ref{fixpoints-t}, that $S \subseteq \mathit{fix}_t \(\star_S\)$ and, consequently, that $|S| \leq |\mathit{fix}_t \(\star_S\)|$. On the other hand, for all $s \in \mathit{fix}_t \(\star_S\)$, $s \in S \cup \bigcup_{t'< t} S_{t'}$ and, consequently, $\mathit{fix}_t \(\\star_S\) \subseteq S \cup \bigcup_{t'< t} S_{t'}$ because, by the way in which $\star_S: U^2 \to U$ was constructed, $\overline{S \cup \bigcup_{t'< t} S_{t'}}$ does not contain fixpoints. Then, we obtain that $|\mathit{fix}_t \(\\star_S\)| \leq |S \cup \bigcup_{t'< t} S_{t'}| = |S|$. Jointly, these two results prove that $|\mathit{fix}_t \(\\star_S\)| = |S|$.
\qed\end{proof}

From the previous result, it is possible to reproduce the result of Prop.~\ref{existsEXP}, but with a minor modification because constructing the algebra requires the use of the set $S_\phi$, instead of the set $\mathit{fix}_t \(\star_{S_\phi}\)$, as there might be fixpoints outside $S_\phi$\footnote{The reader should note the fact that while on the one hand, $S \subseteq \mathit{fix}_t \(\star_S\)$, on the other $\mathit{fix}_t \(\\star_S\) \subseteq S \cup \bigcup_{t'< t} S_{t'}$. Such asymmetry only allows us to guarantee that $S_\phi$ is a set of fixpoints but regarding as possible the existence of fixpoints of $\mathit{fix}_t \(\star_S\)$, which lay outside $S_\phi$.}. Thereafter, Thm.~\ref{existsEXP3} can be applied in order to guarantee the existence of infinitely many prime, big and explosive relation algebras obtained by controlling the fixpoints of $\star: U^2 \to U$ resorting to a term from $\BT$.

Next theorem shows that Prop.~\ref{existsEXPconst} is a special case of Prop.\ref{existsEXPtconst}.

\begin{theorem}{\ \ \ \ \ \ \ \ \ } 
Let $\mathcal{F} = \< F, \pjoin, \pmeet, \pcomple, \pzero, \punit{E}, \pcompo, \pconve, \pid, \pfork \> \in \PFA$, $\mathit{Si}_{\underline{2}} \(\mathcal{F}\) = \mathit{Si}_{\underline{\bin{\nil}{\nil}}} \(\mathcal{F}\)$.
\end{theorem}
\begin{proof}
The proof is direct by using the definitions of $\mathit{Si}_{\underline{2}} \(\mathcal{F}\)$, $\underline{2}$, $\underline{\bin{\nil}{\nil}}$ and $\mathit{Si}_{\underline{\bin{\nil}{\nil}}} \(\mathcal{F}\)$.
\qed\end{proof}

Observing the generalised controlling technique of the fixpoints of $\star$ presented above, it is possible to establish certain relations between them. Let us first introduce some definitions. The next definition generalises $\BT$ by introducing \emph{Binary tree contexts} which, like in rewriting systems such as $\lambda$-calculus, are defined to be binary tree terms with some holes in it, denoted as ``$\mathit{[\ ]}$''.

\newcommand\btc{\mathsf{btc}}
\newcommand\BTC{\mathsf{BTC}}
\newcommand\h{\mathtt{[\ ]}}

\begin{definition}[Binary tree contexts]
\emph{Binary tree contexts} are defined to be $\BTC$, the smallest set of terms produced by the following grammar $\btc ::= \nil\ |\ \h\ |\ \bin{\btc}{\btc}$. 
\end{definition}

\begin{definition}[Substitution]{\ \ \ \ \ \ } 
Let $t \in \BTC$, we define the function $\bullet \mathtt{[}\bullet\mathtt{]}: \BTC \times \BTC \to \BTC$ as follows:
\[
\begin{array}{rcl}
\h[t] & = & t\\
\nil[t] & = & \nil\\
\(\bin{i}{d}\)[t] & = & \bin{\(i[t]\)}{\(d[t]\)} 
\end{array}
\]
\end{definition}

\begin{proposition}
Let $t \in \BTC$ and $t' \in \BT$, then $t[t'] \in \BT$.
\end{proposition}
\begin{proof}
The proof follows easily by induction on $t$.
\qed\end{proof}

\begin{definition}[Variants]{\ \ \ \ \ \ \ } 
Let $t \in \BT$, we define the \emph{variants} of $t$ as $V_t = \setof{t' \in \BTC}{t = t'[\nil]}$.
\end{definition}

\begin{theorem}
Let $\star: U^2 \to U$ and $t, t' \in \BT$ such that $t \neq \nil$ and $t' \neq \nil$, then $\(\forall t'' \in V_{t'}\)\(\mathit{fix}_t \(\star\) \cap \mathit{fix}_{t'} \(\star\) \subseteq \mathit{fix}_{t''[t]}\)$.
\end{theorem}
\begin{proof}
Let $u \in U$ such that $u \in \mathit{fix}_t \(\star\)$ and $u \in \mathit{fix}_{t'} \(\star\)$ then, we know that $\map{t}{\star}{u} = u$ and $\map{t'}{\star}{u} = u$. Let $t'' \in V_{t'}$, then $t''[t]$ is structurally equal to $t'$ with the exception that some of its leaves (those that were $\h$ in $t''$) were replaced by $t$ and in $t'$ are $u$, Then, using that $\map{t}{\star}{u} = u$ and $\map{t'}{\star}{u} = u$, we obtain that $\map{t''[t]}{\star}{u} = u$ and, therefore, $u \in \mathit{fix}_{t''[t]} \(\star\)$.
\qed\end{proof}

\subsection{Controlling the fixpoints of $\star$ through $\pi$ and $\rho$}
In the previous section we presented the generalisation of the technique used by Veloso in \cite{veloso:05-96,veloso:ES-418-96} in the construction of a prime, big and explosive proper relation algebra where explosiveness is guarantied by controlling the cardinality of the set of fixpoints of $\star$, each of which leads to a non-isomorphic fork algebra. In this section we show a similar construction but relying on the quasi-projections $\pi$ and $\rho$.

In the forthcoming paragraph we focus on the use of $\pi$ but it can be reproduced by means of analogous definitions and results for $\rho$.
\begin{definition}[$\underline{\pi}$]{\ \ \ \ \ \ \ } 
\label{underlinepi}
Let $\mathcal{F} = \< A, \pjoin, \pmeet, \pcomple, \pzero, \punit{E}, \pcompo, \pconve, \pid, \pfork \> \in \PFA$, then $\underline{\pi} = \pconv{\pid \pfork \punit{E}}$. 
\end{definition}

\begin{definition}[Subidentities of $\underline{\pi}$]
Let $\mathcal{F} = \< A, \pjoin, \pmeet, \pcomple, \pzero, \punit{E}, \pcompo, \pconve, \pid, \pfork \> \in \PFA$, then $\mathit{Si}_{\underline{\pi}} (\mathcal{F}) = \setof{a \in A}{a \subseteq \underline{\pi} \pmeet \pid}$.
\end{definition}

\begin{proposition}
Let $\mathcal{F}, \mathcal{G} \in \PFA$, if $\phi: \mathcal{F} \to \mathcal{G}$ is an isomorphism, then $\phi$ induces a bijection between $\mathit{Si}_{\underline{\pi}} (\mathcal{F})$ and $\mathit{Si}_{\underline{\pi}} (\mathcal{G})$.
\end{proposition}
\begin{proof}
The proof of this proposition is analogous to that of Prop.~\ref{isomorphism}.
\qed\end{proof}

\begin{definition}[$\pi$-controlled fixpoints of $\star$]
\label{fixpoints-pi}
Let $\star: U^2 \to U$, the fixpoints of $\star$ are defined as $\mathit{fix}_\pi (\star) = \setof{u \in U}{(\exists v \in U)(u \star v = u)}$.
\end{definition}

This set can also be presented as a relation contained in the identity relation, as follows:
\begin{equation}
\label{Id-fixpoints-pi}
\pid_{\mathit{fix}_\pi (\star)} = \setof{\langle u, u \rangle \in U^2}{u \in \mathit{fix}_\pi (\star)}
\end{equation}
\noindent for which it is possible to prove the following properties.

\begin{proposition}{\ \ \ \ \ \ \ \ \ } 
Let $\mathcal{F} = \< A, \pjoin, \pmeet, \pcomple, \pzero, \punit{E}, \pcompo, \pconve, \pid, \pfork \> \in \PFA$, then $\underline{\pi} \pmeet \pid = \pid_{\mathit{fix}_\pi (\star)}$.
\end{proposition}
\begin{proof}
\[
\begin{array}{rcl}
\<u, v\> \in \underline{\pi} \pmeet \pid & \text{iff} & \<u, v\> \in \pconv{\pid \pfork \punit{E}} \text{ and } u = v \\
                                                            && \qquad \text{[by Defs.~\ref{underlinepi} and~\ref{def_proper-relation-algebra} - $\pid$.]}\\
                                                            & \text{iff} & \<v, u\> \in \pid \pfork \punit{E} \text{ and } u = v \\
                                                            && \qquad \text{[by Def.~\ref{def_proper-relation-algebra} - $\pconve$.]}\\
                                                            & \text{iff} & \text{there exist } r, s \in U \text{ such that } u = r \star s \text{, } \<v, r\> \in \pid \text{, }\\
                                                            &              & \quad \<v, s\> \in \punit{E} \text{ and } u = v \\
                                                            && \qquad \text{[by Def.~\ref{def_proper-fork-algebras} - $\pfork$.]}\\
                                                            & \text{iff} & \text{there exist } r, s \in U \text{ such that } u = r \star s \text{, } v = r \text{ and } u = v \\
                                                            && \qquad \text{[by Def.~\ref{def_proper-relation-algebra} - $\pid$ and $\punit{E}$.]}\\
                                                            & \text{iff} & \text{there exists } s \in U \text{ such that } u = v \star s \text{ and } u = v \\
                                                            & \text{iff} & \text{there exists } s \in U \text{ such that } u = u \star s \text{ and } u = v \\
                                                            & \text{iff} & u \in \setof{u' \in U}{(\exists s \in U)(u' = u' \star s)} \text{ and } u = v \\
                                                            & \text{iff} & u \in \mathit{fix}_\pi (\star) \text{ and } u = v \\
                                                            && \qquad \text{[by Def.~\ref{fixpoints-pi}.]}\\
                                                            & \text{iff} & \<u, u\> \in \pid_{\mathit{fix}_\pi (\star)} \text{ and } u = v \\
                                                            && \qquad \text{[by Eq.~\ref{Id-fixpoints-pi}.]}\\
                                                            & \text{iff} & \<u, v\> \in \pid_{\mathit{fix}_\pi (\star)}
\end{array}
\]
Thus, finishing the proof.\qed\end{proof}

\begin{proposition}
Let $\mathcal{F} = \< A, \pjoin, \pmeet, \pcomple, \pzero, \punit{E}, \pcompo, \pconve, \pid, \pfork \> \in \PFA$ simple, such that $\pfork$ is induced by $\star: U^2 \to U$, then $\mathit{Si}_{\underline{2}} (\mathcal{F}) = \wp (\pid_{\mathit{fix} (\star)}) \pmeet A$.
\end{proposition}
\begin{proof}
The proof of this proposition is analogous to that of Prop.~\ref{Sit-powerset}.
\qed\end{proof}

The following property is analogous to Prop.~\ref{existsEXPconst} but resorting to $\pi$ as a determinant means for controlling the fixpoints of $\star$.

\begin{proposition}
\label{existsEXPpiconst}
Let $U$ be an infinite set such that $|U| = \kappa$ and $\aleph_0 \leq \kappa$ then, for all $S \subseteq U$ such that $|S| < |U|$, there exists $\star_S: U^2 \to U$ injective such that $\mathit{fix}_ \pi \(\star_S\) = S$.
\end{proposition}
\begin{proof}
If $|U| = \kappa$, then $|U^2| = \kappa$, $|\pid| = \kappa$ and $|\overline{Id}| = \kappa$. Then, let $S \subseteq U$ such that $|S| < \kappa$, we know that $|\overline{S}| = \kappa$ and, therefore, it is possible to take $\overline{S} = A \cup \bigcup_{i \in \nat} B_i$ such that:
\begin{itemize}
\item for all $i \in \nat$,$A \cap B_i = \emptyset$,
\item for all $i, j \in \nat$,such that $i \not= j$, $B_i \cap B_j = \emptyset$, 
\item $|A| = \kappa$ and for all $i \in \nat$, $|B_i| = \kappa$.
\end{itemize}
Then, there exists $g: \overline{S} \to \bigcup_{i \in \nat} B_i$ bijective and without fixpoints defined as the union of the bijective functions $g_A: A \to B_0$ and $\{g_i: B_i \to B_{i+1}\}_{i \in \nat}$. Notice that such union is disjoint as the domains and codomains of all of the functions are disjoint.

Let $P$ be a set such that $|P| = |S|$ and $l: S \to P$ be a bijective function. Then, there exists $f: (S \cup P)^2 \to A$ injective and, consequently, it is possible to define $\star_S: U^2 \to U$  according to Table~\ref{table2}
\begin{table}[ht]
\centering
\begin{tabular}{|c||c|c|c|c|}
\hline
$u \star_S v$ & $S$ & $P$ & $A$ & $B_j$ \\
\hline
\hline
$S$                 & $f (u. v)$ & $\left\{ \begin{array}{ll}
                                                    u         & \mbox{; if $v = l (u)$.}\\
                                                    f (u, v) & \mbox{; otherwise.}
                                                \end{array}\right.$                       & $g(v)$ & $g (v)$\\
\hline
P                     & $f (u. v)$ & $f (u. v)$ & $g(v)$ & $g (v)$\\
\hline
A                     & $g (u)$    & $g (u)$   & $g(u)$ & $g (v)$\\
\hline
$B_j$              & $g (u)$    & $g (u)$   & $g(u)$ & $\left\{ \begin{array}{ll}
                                                                                     g (u) & \mbox{; if $i \geq j$.}\\
                                                                                     g (v) & \mbox{; otherwise.}
                                                                                 \end{array}\right.$\\
\hline
\end{tabular}
\caption{Definition of $\star_S: U^2 \to U$ controlled by $\pi$.}
\label{table2}
\end{table}
                   
Table~\ref{table2} shows how $\star_S: U^2 \to U$ is defined as the union of injective functions with disjoint domains and codomains. Therefore, $\star_S: U^2 \to U$ is injective an injective function such that $\mathit{fix}(\star_S) = S$.
\qed\end{proof}

Once again, it is possible to apply Prop.~\ref{existsEXP} and Thm.~\ref{existsEXP3} in order to prove the existence of many prime, big and explosive proper relation algebras.\\

As we mentioned at the beginning of this section, analogous definitions and results can be developed for controlling the fixpoints of $\star$ but resorting to $\rho$.

\subsection{Generalising the control of the fixpoints of $\star$ through the projections $\pi$ and $\rho$}
The generalisation of the controlling technique of the fixpoints of $\star$ through $\pi$ and $\rho$ presented in the previous section is somehow similar to what was presented in the previous section.

Let us first consider the following data type formalising non-empty sequences of relations ``$\pi$'' and ``$\rho$''.

\newcommand\secu{\mathsf{sec}}
\newcommand\Secu{\mathsf{Sec}}
\newcommand\elem[1]{\mathtt{elem}\ #1}
\newcommand\cons[2]{\mathtt{cons}\ #1\ #2}

\begin{definition}[Sequences]
\emph{Sequences} are the elements of $\Secu$, the smallest set of terms produced by the following grammar $\secu ::= \elem{*}\ |\ \cons{*}{s}$, where $* \in \set{\pi, \rho}$ and $s \in \Secu$.
\end{definition}

\begin{definition}
\label{functions}
The functions $\mathit{long}: \Secu \to \nat$, $\bullet[\bullet]: \Secu \times \nat \to \set{\pi, \rho}$ and $\bullet|\bullet: \Secu \times\nat \to \Secu$\footnote{Note that the last two functions are partial and are only defined on those elements $n \in \nat$ and $s \in \Secu$ such that $1 \leq n \leq \mathit{long} (s)$.} are defined as follows: let $* \in \set{\pi, \rho}$ and $s \in \Secu$,
\[
\begin{array}{rcl}
\mathit{long} (\elem{*}) & = & 1\\
\mathit{long} (\cons{*}{s}) & = & 1 + \mathit{long} (s)
\end{array}
\]
\[
\begin{array}{rcl}
(\elem{*})[1] & = & \elem{*}\\
(\cons{*}{s})[i] & = & \left\{\begin{array}{lr}* & \text{; if $i = 1$.} \\ s[i-1] & \text{; otherwise.} \end{array} \right.
\end{array}
\]
\[
\begin{array}{rcl}
(\elem{*})|1 & = & *\\
(\cons{*}{s})|i & = & \left\{\begin{array}{lr}(\cons{*}{s}) & \text{; if $i = \mathit{long} (s) + 1$.} \\ s|i & \text{; otherwise.} \end{array} \right.
\end{array}
\]
\end{definition}

\begin{definition}
\label{sec-underline}
Let $\mathcal{F} = \< A, \pjoin, \pmeet, \pcomple, \pzero, \punit{E}, \pcompo, \pconve, \pid, \pfork \> \in \PFA$, $* \in \set{\pi, \rho}$ and $s \in \Secu$, then $\underline{\elem{*}} = \underline{*}$ and $\underline{\cons{*}{s}} = \underline{*} \pcompo \underline{s}$.
\end{definition}

\begin{definition}[Subidentities of $\underline{s}$]
\label{def:sis}
Let $\mathcal{F} = \< A, \pjoin, \pmeet, \pcomple, \pzero, \punit{E}, \pcompo, \pconve, \pid, \pfork \> \in \PFA$ and $s \in \Secu$, then $\mathit{Si}_{\underline{s}} (\mathcal{F}) = \setof{a \in A}{a \subseteq \underline{s} \pmeet \pid}$.
\end{definition}

\begin{proposition}
\label{isomorphisms}
Let $\mathcal{F} = \< F, \pjoin^\mathcal{F}, \pmeet^\mathcal{F}, \pcomple^\mathcal{F}, \pzero^\mathcal{F}, \punit{E}^\mathcal{F}, \pcompo^\mathcal{F}, {\pconve}^\mathcal{F}, \pid^\mathcal{F}, \pfork^\mathcal{F} \>$ and $\mathcal{G} = \< G, \pjoin^\mathcal{G}, \pmeet^\mathcal{G}, \pcomple^\mathcal{G}, \pzero^\mathcal{G}, \punit{E}^\mathcal{G}, \pcompo^\mathcal{G}, {\pconve}^\mathcal{G}, \pid^\mathcal{G}, \pfork^\mathcal{G} \>$ be $\PFA$ and $s \in \Secu$, if $\phi: \mathcal{F} \to \mathcal{G}$ is an isomorphism, then $\phi$ induces a bijection between $\mathit{Si}_{\underline{s}^\mathcal{F}} (F)$ and $\mathit{Si}_{\underline{s}^\mathcal{G}} (G)$.
\end{proposition}
\begin{proof}
Let $\phi: F \to G$ be an isomorphism between $\mathcal{F}$ and $\mathcal{G}$, and $r \in F$ such that $r \in \mathit{Si}_{\underline{s}^\mathcal{F}} (F)$,
\[
\begin{array}{rclr}
r \in \mathit{Si}_{\underline{s}^\mathcal{F}} (F) 
		& \text{ iff } & r \subseteq^\mathcal{F} \underline{s}^\mathcal{F} \pmeet^\mathcal{F} \pid^\mathcal{F} & \text{[by Def.~\ref{def:sis}.]}\\
		& \text{ iff } & r \pjoin^\mathcal{F} \(\underline{s}^\mathcal{F} \pmeet^\mathcal{F} \pid^\mathcal{F}\) = \underline{s}^\mathcal{F} \pmeet^\mathcal{F} \pid^\mathcal{F} & \text{[by Def. of $\subseteq$.]}\\
		& \text{ iff } & \phi\(r \pjoin^\mathcal{F} \(\underline{s}^\mathcal{F} \pmeet^\mathcal{F} \pid^\mathcal{F}\)\) = \phi\(\underline{s}^\mathcal{F} \pmeet^\mathcal{F} \pid^\mathcal{F}\) \\
		&&\qquad \text{[because $\phi$ is an isomorphism.]}\\
		& \text{ iff } & \phi\(r\) \pjoin^\mathcal{G} \(\phi\(\underline{s}^\mathcal{F}\) \pmeet^\mathcal{G} \pid^\mathcal{G}\) = \phi\(\underline{s}^\mathcal{F}\) \pmeet^\mathcal{G} \pid^\mathcal{G} \\
		&&\qquad \text{[because $\phi$ is an isomorphism.]}\\
		& \text{ iff } & \phi\(r\) \pjoin^\mathcal{G} \(\underline{s}^\mathcal{G} \pmeet^\mathcal{G} \pid^\mathcal{G}\) = \underline{s}^\mathcal{G} \pmeet^\mathcal{G} \pid^\mathcal{G} & \text{[by Lemma~\ref{bijection-s}.]}\\
		& \text{ iff } & \phi\(r\) \subseteq^\mathcal{G} \underline{s}^\mathcal{G} \pmeet^\mathcal{G} \pid^\mathcal{G} & \text{[by Def. of $\subseteq$.]}\\
		& \text{ iff } & \phi\(r\) \in \mathit{Si}_{\underline{s}^\mathcal{G}} (G) & \text{[by Def.~\ref{def:sis}.]}
\end{array}
\]
Thus, finishing the proof.\qed\end{proof}

\begin{definition}[$s$-controlled fixpoints of $\star$]
\label{fixpoints-s}
Let $\star: U^2 \to U$ and $s \in \Secu$, the $s$-controlled fixpoints of $\star$ are defined as $\mathit{fix}_s (\star) = \setof{u \in U}{\<u, u\> \in \underline{s}}$.
\end{definition}

Recalling the definition of the set of fixpoints of $\star$ given in Def.~\ref{fixpoints-s} we can present the $s$-controlled fixpoints of $\star$ as a partial identity as follows.
\begin{equation}
\label{Id-fixpoints-s}
\pid_{\mathit{fix}_s (\star)} = \setof{\langle u, u \rangle \in U^2}{u \in \mathit{fix}_s (\star)}
\end{equation}
\noindent for which it is possible to derive the following properties.

\begin{proposition}
\label{prop-id-s}
Let $\mathcal{F} = \< F, \pjoin, \pmeet, \pcomple, \pzero, \punit{E}, \pcompo, \pconve, \pid, \pfork \> \in \PFA$ with $\pfork$ induced by $\star: U^2 \to U$ and $s \in \Secu$, then $\underline{s} \pmeet \pid = \pid_{\mathit{fix}_s (\star)}$.
\end{proposition}
\begin{proof}
The proof is analogous to that of Prop.~\ref{prop-id}.
\qed\end{proof}

\begin{proposition}
\label{Sis-powerset}
Let $\mathcal{F} = \< F, \pjoin, \pmeet, \pcomple, \pzero, \punit{E}, \pcompo, \pconve, \pid, \pfork \> \in \PFA$ simple with $\pfork$ induced by $\star: U^2 \to U$ and $s \in \Secu$, then $\mathit{Si}_{\underline{s}} (\mathcal{F}) = \wp \(\pid_{\mathit{fix}_s (\star)}\) \pmeet F$.
\end{proposition}
\begin{proof}
The proof is analogous to that of Prop.~\ref{Sit-powerset}.
\qed\end{proof}

\begin{proposition}
\label{existsEXPsconst}
Let $U$ be an infinite set $|U| = \kappa$ and $\aleph_0 \leq \kappa$ then for all $S \subseteq U$ such that $|S| < |U|$, there exists $\star_S: U^2 \to U$ injective such that given $s \in \Secu$, $\left|\mathit{fix}_s \(\star_S\)\right| = |S|$.
\end{proposition}
\begin{proof}
If $|U| = \kappa$, then $|U^2| = \kappa$, $|\pid| = \kappa$ and $|\overline{Id}| = \kappa$. Then, let $\{S_i\}_{1 \leq i < \mathit{long} (s)}$ be a finite family of sets such that:
\begin{itemize}
\item for all $1 \leq i, j < \mathit{long} (s)$, such that $i \not= j$, $B_i \cap B_j = \emptyset$, and
\item for all $1 \leq i < \mathit{long} (s)$, $|B_i| = |S|$.
\end{itemize}
Analogous to previous results, we know that $|\overline{S \cup \bigcup_{i = 1}^{\mathit{long} (s)} S_{i}}| = \kappa$ and, therefore, it is possible to take $\overline{S \cup \bigcup_{i = 1}^{\mathit{long} (s)} S_{i}} = \bigcup_{i \in \nat} B_i$ such that:
\begin{itemize}
\item for all $i, j \in \nat$, such that $i \not= j$, $B_i \cap B_j = \emptyset$, 
\item for all $i \in \nat$, $|B_i| = \kappa$.
\end{itemize}
Then, there exists $g: \bigcup_{i \in \nat} B_i \to \bigcup_{i \in \nat} B_i$ bijective and without fixpoints defined as the union of the bijective functions $\{g_i: B_i \to B_{i+1}\}_{i \in \nat}$. Notice that such union is disjoint as the domains and codomains of all of the functions are disjoint.

Let $P$ be a set such that $|P| = |S|$ and $l: S \to P$ be a bijective function. Then, there exists $f: (S \cup P)^2 \to B_0$ injective and a finite family of bijective functions $\{h^s_{i}: S \to S_{i}\}_{1 \leq i < \mathit{long} (s)}$, such that $h^s_0 = id_S$. Then, it is possible to define $\star_S: U^2 \to U$  according to Table~\ref{table3}
\begin{sidewaystable}
\centering
\begin{tabular}{|c||c|c|c|c|}
\hline
$u \star_S v$ & $S_{j}$ & $P$ & $B_n$ \\
\hline
\hline
$S_{i}$ & f (u, v) & $\left\{ \begin{array}{ll} 
                                           h^s_{i+1} \({h^s_{i}}^{-1} (u)\)      & \begin{array}{ll} \mbox{; if $1 \leq i < \mathit{long} (s)$}\\ \qquad\mbox{and $s [\mathit{long} (s) - i] = \pi$.}\\ \qquad\mbox{and $v = l \({h^s_i}^{-1} (u)\)$.} \end{array} \\
                                          {h^s_{i}}^{-1} (u)                           & \begin{array}{ll} \mbox{; if $i = \mathit{long} (s) - 1$}\\ \qquad\mbox{and $s [1] = \pi$.}\\ \qquad\mbox{and $v = l \({h^s_i}^{-1} (u)\)$.} \end{array} \\
                                          f \( {h^s_{i}}^{-1} (u), v \)  & \mbox{; otherwise.}
                                \end{array} \right.$ & $g(v)$ \\
\hline
$P$      & $\left\{ \begin{array}{ll} 
                            h^s_{i+1} \({h^s_{i}}^{-1} (v)\)       & \begin{array}{ll} \mbox{; if $1 \leq i < \mathit{long} (s)$}\\ \qquad\mbox{and $s [\mathit{long} (s) - i] = \rho$.}\\ \qquad\mbox{and $u = l \({h^s_i}^{-1} (v)\)$.} \end{array} \\
                            {h^s_{i}}^{-1} (v)                           & \begin{array}{ll} \mbox{; if $i = \mathit{long} (s) - 1$}\\ \qquad\mbox{and $s [1] = \rho$.}\\ \qquad\mbox{and $u = l \({h^s_i}^{-1} (v)\)$.} \end{array} \\
                            f \( u, {h^s_{i}}^{-1} (v) \)  & \mbox{; otherwise.}
                  \end{array} \right.$ & $f (u, v)$ & $g(v)$ \\
\hline
$B_m$  & $g (u)$  & $g (u)$ & $\left\{ \begin{array}{ll}
                                                             g (u) & \mbox{; if $m \geq n$.}\\
                                                             g (v) & \mbox{; otherwise.}
                                                    \end{array}\right.$ \\
\hline
\end{tabular}
\caption{Definition of $\star_S: U^2 \to U$ controlled by a non empty sequence.}
\label{table3}
\end{sidewaystable}

Table~\ref{table3} shows how $\star_S: U^2 \to U$ is defined as the union of injective functions with disjoint domains and codomains. Therefore, $\star_S: U^2 \to U$ is injective.

On the one hand, by Lemma~\ref{Sfixpoint-s} we obtain that if $u \in S$, $\langle u, u \rangle \in \underline{s}$ and, by Def.~\ref{fixpoints-s}, that $S \subseteq \mathit{fix}_s \(\star_S\)$ and, consequently, that $|S| \leq |\mathit{fix}_s \(\star_S\)|$. On the other hand, for all $u \in \mathit{fix}_s \(\star_S\)$, $u \in S \cup \bigcup_{i = 1}^{\mathit{long} (s)} S_{i}$ and, consequently, $\mathit{fix}_s \(\\star_S\) \subseteq S \cup \bigcup_{i = 1}^{\mathit{long} (s)} S_{i}$ because, by the way in which $\star_S: U^2 \to U$ was constructed, $\overline{S \cup \bigcup_{i = 1}^{\mathit{long} (s)} S_{i}}$ does not contain fixpoints. Then, we obtain that $|\mathit{fix}_s \(\\star_S\)| \leq |S \cup \bigcup_{i = 1}^{\mathit{long} (s)} S_{i}| = |S|$. Jointly, these two results prove that $|\mathit{fix}_s \(\\star_S\)| = |S|$.
\qed\end{proof}

Once again, from the previous result, it is possible to reproduce the result of Prop.~\ref{existsEXP}. Analogous to what we did in the proof of Prop.~\ref{existsEXPtconst}, we must consider the use of the set $S_\phi$, instead of the set $\mathit{fix}_t \(\star_{S_\phi}\)$, for constructing the algebra, as there might be fixpoints outside $S_\phi$. Thereafter, Thm.~\ref{existsEXP3} can be applied in order to guarantee the existence of infinitely many prime, big and explosive relation algebras obtained by controlling the fixpoints of $\star: U^2 \to U$ resorting to a sequence from $\Secu$.\\

Once again, from observing the generalised controlling technique of the fixpoints of $\star$ presented above, it is possible to establish certain relations between them.

\newcommand\concat[2]{#1 +\!\!+\ #2}

\begin{definition}
The functions $\concat{\bullet}{\bullet}: \Secu^2 \to \Secu$ is defined as follows: let $* \in \set{\pi, \rho}$ and $s, s' \in \Secu$
\[
\begin{array}{rcl}
\concat{\(\elem{*}\)}{s} & = & \cons{*}{s}\\
\concat{\(\cons{*}{s'}\)}{s} & = & \cons{*}{\(\concat{s'}{s}\)}
\end{array}
\]
\end{definition}

\begin{theorem}
Let $\star: U^2 \to U$ and $s, s' \in \Secu$, $\mathit{fix}_s (\star) \cap \mathit{fix}_{s'} (\star) \subseteq \mathit{fix}_{\concat{s}{s'}}$.
\end{theorem}
\begin{proof}
\[
\begin{array}{rclr}
u \in \mathit{fix}_s (\star) \cap \mathit{fix}_{s'} (\star) & \text{iff} & u \in \mathit{fix}_s (\star) \text{ and } \mathit{fix}_{s'} (\star) \\
                             & \text{iff}         & \< u, u \> \in \underline{s} \text{ and } \< u, u \> \in \underline{s'} & \text{[by Def.~\ref{fixpoints-s}.]}\\
                             & \text{implies} & \< u, u \> \in \underline{s} \pcompo \underline{s'} & \text{[by Def.~\ref{def_proper-relation-algebra} - $\pcompo$.]}\\
                             & \text{iff}         & \< u, u \> \in \underline{\concat{s}{s'}} & \text{[by Lemma~\ref{fix-inclusion}.]}\\
                             & \text{iff}         & u \in \mathit{fix}_{\concat{s}{s'}} & \text{[by Def.~\ref{fixpoints-s}.]}
\end{array}
\]
Thus, finishing the proof.\qed\end{proof}

Finally, it is possible to connect $t$-controled fixpoints and $s$-controlled fixpoints of $\star$ by considering properties like the next one.

\begin{definition}
\label{function-ll}
The predicates $\bullet = \bullet \subseteq \BT \times \BT$ and $\bullet < \bullet \subseteq \BT \times \BT$ are defined as follows:
\[
\begin{array}{rcl}
s \ll t & \text{iff} & \(s = \elem{\pi} \text{ and } t = \bin{\nil}{t'}\) \text{ or}\\
         &             & \(s = \elem{\rho} \text{ and } t = \bin{t'}{\nil}\) \text{ or}\\
         &             & \(s = \cons{\pi}{s'} \text{ and } t = \bin{t'}{t''} \text{ and } s' \ll t'\) \text{ or}\\
         &             & \(s = \cons{\rho}{s'} \text{ and } t = \bin{t'}{t''} \text{ and } s' \ll t''\)
\end{array}
\]
\end{definition}

\begin{theorem}
Let $U$ be an infinite set $|U| = \kappa$ with $\aleph_0 \leq \kappa$, $S \subseteq U$ such that $|S| < |U|$, $\star_S: U^2 \to U$ injective and $t \in \BT$ then, for all $s \in \Secu$, $s \ll t$ implies $\mathit{fix}_t (\star_S) \subseteq \mathit{fix}_{s} (\star_S)$.
\end{theorem}
\begin{proof}
The proof of this theorem follows from Def.~\ref{fixpoints-t}, \ref{fixpoints-s} and~\ref{function-ll}, and applying Lemma~\ref{fix-t-fix-s}.
\qed\end{proof}
\section{Conclusions}
\label{conclusions}

As we mentioned at the beginning of this work, binary relations are ubiquitous in computer science as they provide the concept perfectly fitted for formalising programs by rationalising them as the connection between its inputs and its outputs. In this context, the language of relation algebras is expected to provide the reasoning tool for program verification, derivation and refinement. The mismatch between the models of the calculus of relations (see Defs.\ref{def_relation-algebra} and~\ref{def_new-calculus-of-relations}) and the class of proper relation algebras (see Def.~\ref{def_proper-relation-algebra}), evidenced by Lyndon in \cite{lyndon:ams2-51_2,lyndon:ams2-63_2}, by constructing a finite, non-simple and non-trivial relation algebra that is not representable as a proper relation algebra, results in a major drawback for its adoption as a specification language and formal development tool.

The study of the relational reduct of fork algebras, started and promoted by Paulo {A.}{S.} Veloso in \cite{veloso:ES-418-96,veloso:05-96}, is of great interest for the community of applied relational methods in computer science as fork algebras, thought of as the models of the calculus for fork algebras (see Defs.~\ref{def_fork-algebras} and~\ref{def_calculus-of-fork-algebras}), are representable in proper fork algebras (see Def.~\ref{def_proper-fork-algebras}), a class of algebras whose carrier is formed by binary relations. 

In this paper we summarised some of Velosos's results in this field, like the construction of explosive relation algebras, by controlling the fixpoints of $\star: U^2 \to U$. Our contribution is twofold; on the one hand, a generalisation of such a construction by introducing the notion of $t$-controlled fixpoints of $\star: U^2 \to U$, where $t$ is a term induced by a tree-like structure and, on the other hand, the controlling technique based on the use of the pseudo-projections $\pi$ and $\rho$, as an alternative to the $\pfork$-controlled one, introduced by Veloso. Finally, we generalise the technique by introducing the notion of $s$-controlled fixpoints of $\star: U^2 \to U$, where $s$ is a term induced by a sequence-like structure.

\bibliographystyle{splncs04}
\bibliography{bibdatabase}

\appendix
\section{Proofs for selected lemmas and properties}
\label{demos}
In this section we present selected auxiliary lemmas and properties used in the preceding sections.

\begin{lemma}{\ \ \ \ \ \ \ \ \ } 
\label{bijection-t}
Let $\mathcal{F} = \< F, \pjoin^\mathcal{F}, \pmeet^\mathcal{F}, \pcomple^\mathcal{F}, \pzero^\mathcal{F}, \punit{E}^\mathcal{F}, \pcompo^\mathcal{F}, {\pconve}^\mathcal{F}, \pid^\mathcal{F}, \pfork^\mathcal{F} \>$ and $\mathcal{G} = \< G, \pjoin^\mathcal{G}, \pmeet^\mathcal{G}, \pcomple^\mathcal{G}, \pzero^\mathcal{G}, \punit{E}^\mathcal{G}, \pcompo^\mathcal{G}, {\pconve}^\mathcal{G}, \pid^\mathcal{G}, \pfork^\mathcal{G} \>$ be $\PFA$ and $t \in \BT$, if $\phi: \mathcal{F} \to \mathcal{G}$ is an isomorphism, then $\phi \(\underline{t}^\mathcal{F}\) = \underline{t}^\mathcal{G}$.
\end{lemma}
\begin{proof}
The proof of this lemma follows by induction on the structure of $t$. Let us consider $t = \nil$ as the base case.
\[
\begin{array}{rclr}
\phi \(\underline{\nil}^\mathcal{F}\) & = & \phi \(\map{\nil}{\pfork^\mathcal{F}}{\pid^\mathcal{F}}\) & \text{[by Def.~\ref{term-underline}.]}\\
		& = & \phi \(\pid^\mathcal{F}\) & \text{[by Def.~\ref{map}.]}\\
		& = & \pid^\mathcal{G} & \text{[because $\phi$ is isomorphism.]}\\
		& = & \map{\nil}{\pfork^\mathcal{G}}{\pid^\mathcal{G}} & \text{[by Def.~\ref{map}.]}\\
		& = & \underline{\nil}^\mathcal{F} & \text{[by Def.~\ref{term-underline}.]}
\end{array}
\]
Let us now consider the case $t = \bin{i}{d}$.
\[
\begin{array}{rclr}
\phi \(\underline{\bin{i}{d}}^\mathcal{F}\) & = & \phi \(\map{\(\bin{i}{d}\)}{\pfork^\mathcal{F}}{\pid^\mathcal{F}}\) & \text{[by Def.~\ref{term-underline}.]}\\
		& = & \phi \(\(\map{i}{\pfork^\mathcal{F}}{\pid^\mathcal{F}}\) \pfork^\mathcal{F} \(\map{d}{\pfork^\mathcal{F}}{\pid^\mathcal{F}}\)\) & \text{[by Def.~\ref{map}.]}\\
		& = & \phi \(\map{i}{\pfork^\mathcal{F}}{\pid^\mathcal{F}}\) \pfork^\mathcal{G} \phi \(\map{d}{\pfork^\mathcal{F}}{\pid^\mathcal{F}}\) \\
		&&\qquad \text{[because $\phi$ is isomorphism.]}\\
		& = & \phi \(\underline{i}^\mathcal{F}\) \pfork^\mathcal{G} \phi \(\underline{d}^\mathcal{F}\) & \text{[by Def.~\ref{term-underline}.]}\\
		& = & \underline{i}^\mathcal{G} \pfork^\mathcal{G} \underline{d}^\mathcal{G} \\
		&&\qquad \text{[by inductive hypothesis.]}\\
		& = & \(\map{i}{\pfork^\mathcal{G}}{\pid^\mathcal{G}}\) \pfork^\mathcal{G} \(\map{d}{\pfork^\mathcal{G}}{\pid^\mathcal{G}}\) & \text{[by Def.~\ref{term-underline}.]}\\
		& = & \map{\(\bin{i}{d}\)}{\pfork^\mathcal{G}}{\pid^\mathcal{G}} & \text{[by Def.~\ref{map}.]}
\end{array}
\]
Thus, finishing the proof.\qed\end{proof}

\begin{lemma}
\label{fork2star}
Let $\mathcal{F} = \< F, \pjoin, \pmeet, \pcomple, \pzero, \punit{E}, \pcompo, \pconve, \pid, \pfork \> \in \PFA$ with $\pfork$ induced by $\star: U^2 \to U$ and $t \in \BT$ such that $t \not= \nil$, then $\< u, v \> \in \map{t}{\pfork}{\pid}$ in and only if $v = \map{t}{\star}{u}$.
\end{lemma}
\begin{proof}
The proof of this lemma follows by induction on the structure of $t$. Let us consider $t = \nil$ as the base case.
\[
\begin{array}{rcll}
\< u, v \> \in \map{\nil}{\pfork}{\pid} & \text{ iff } & \< u, v \> \in \pid & \text{[by Def.~\ref{map}.]}\\
	& \text{ iff } & u = v & \text{[by Def.~$\pid$.]}\\
	& \text{ iff } & v = \map{t}{\star}{u} & \text{[by Def.~\ref{map}.]}
\end{array}
\]
Let us now consider the case $t = \bin{i}{d}$.
\[
\begin{array}{clr}
    & \< u, v \> \in \map{\(\bin{i}{d}\)}{\pfork}{\pid}\\
 \text{ iff } & \< u, v \> \in \(\map{i}{\pfork}{\pid}\) \pfork \(\map{d}{\pfork}{\pid}\) & \text{[by Def.~\ref{map}.]}\\
 \text{ iff } & \text{there exists } v', v'' \in U \text{ such that } v = v' \star v'' \text{ and } \\
    & \<u, v'\> \in \map{i}{\pfork}{\pid} \text{ and } \<u, v''\> \in \map{d}{\pfork}{\pid} & \text{[by Def. $\pfork$.]}\\
 \text{ iff } & \text{there exists } v', v'' \in U \text{ such that } v = v' \star v'' \text{ and } \\
    & v' = \map{i}{\star}{u} \text{ and } v'' = \map{d}{\star}{u}\\
    & \qquad \text{[by inductive hypothesis.]}\\
 \text{ iff } & v = \(\map{i}{\star}{u}\) \star \(\map{d}{\star}{u}\) \\
    & \qquad \text{[by inductive hypothesis.]}\\
 \text{ iff } & v = \map{\(\bin{i}{d}\)}{\star}{u} & \text{[by Def.~\ref{map}.]}\\
\end{array}
\]
Thus, finishing the proof.\qed\end{proof}

\begin{lemma}
\label{Sfixpoint}
Let $U$ and $S \in U$ be a non-empty sets, $t \in \BT$ such that $t \neq \nil$ and $\star_S: U^2 \to U$ be the injective function defined according to Prop.~\ref{existsEXPtconst}, then $v \in S$ implies that $v = \map{t}{\star_S}{v}$. 
\end{lemma}
\begin{proof}
Let $t = \bin{t'}{t''}$,
\[
\begin{array}{rclr}
v \in S & \text{implies} & h_{t'} (v) \in S_{t'} \text{ and } h_{t''} (v) \in S_{t''} \\
           &                      &\qquad \text{[by Def.~$h_{t}$ for $t \in \BT$.]}\\
           & \text{implies} & h_{t'} (v) = \map{t'}{\star_S}{{h_{t'}}^{-1} \(h_{t'} (v)\)} \text{ and }\\
           &                      & \qquad h_{t''} (v) = \map{t''}{\star_S}{{h_{t''}}^{-1} \(h_{t''} (v)\)} & \text{[by Lemma~\ref{aux4Sfixpoint}.]}\\
           & \text{implies} & h_{t'} (v) = \map{t'}{\star_S}{v} \text{ and }\\
           &                      & \qquad h_{t''} (v) = \map{t''}{\star_S}{v} \\
           &                      & \qquad\qquad \text{[by Def.~$h_{t}$ for $t \in \BT$.]}\\
           & \text{implies} & h_{t'} (v) \star_S h_{t''} (v) =  \\
           &                      &  \qquad\(\map{t'}{\star_S}{v}\) \star_S \(\map{t''}{\star_S}{v}\) & \text{[by Def.~\ref{map}.]}\\
           & \text{implies} & h^{-1}_{t'} (\(h_{t'} (v)\) = \map{\(\bin{t'}{t''}\)}{\star_S}{v} & \text{[by Def.~$\star_S$.]}\\
           & \text{implies} & v = \map{\(\bin{t'}{t''}\)}{\star_S}{v} & \text{[by Def.~$h_t$.]}
\end{array}
\]
Thus, finishing the proof.\qed\end{proof}

\begin{lemma}
\label{aux4Sfixpoint}
Let $U$ and $S \in U$ be a non-empty sets, $t \in \BT$ such that $t \neq \nil$ and $\star_S: U^2 \to U$ be the injective function defined according to Prop.~\ref{existsEXPtconst}, then for all $t'\in \BT$ such that $t' < t$, $v \in S_ {t'}$ implies $v = \map{t'}{\star_S}{{h_{t'}}^{-1} (v)}$.
\end{lemma}
\begin{proof}
The proof follows by induction on the structure of $t$. Let $t' = \nil$, if $v \in S$ then, by definition of $h_\nil$, $v = {h_\nil}^{-1} \(v\)$ and, therefore, by Def.~\ref{map}, $v = \map{\nil}{\star_S}{{h_\nil}^{-1} (v)}$.

Let $t' = \bin{i}{d}$; if $v \in S_ {\bin{i}{d}}$, then there exists $u \in S_i$ and $u' \in S_d$ such that ${h_i}^{-1} (u) = h^{-1}_d (u')$ and $v = u \star_S u'$. By inductive hypothesis, we get that $u = \map{i}{\star_S}{{h_i}^{-1} (u)}$ and $u' = \map{d}{\star_S}{{h_d}^{-1} (u')}$. Then, by replacing in the expression of $v$, $v = \(\map{i}{\star_S}{{h_i}^{-1} (u)}\) \star_S \(\map{d}{\star_S}{{h_d}^{-1} (u')}\)$. Finally, by Def.~\ref{map}, we obtain that $v = \map{\(\bin{i}{d}\)}{\star_S}{{h_i}^{-1} (u)}$ and, equivalently, $v = \map{t'}{\star_S}{{h_i}^{-1} (u)}$.
\qed\end{proof}

\begin{lemma}{\ \ \ \ \ \ \ \ \ } 
\label{bijection-s}
Let $\mathcal{F} = \< F, \pjoin^\mathcal{F}, \pmeet^\mathcal{F}, \pcomple^\mathcal{F}, \pzero^\mathcal{F}, \punit{E}^\mathcal{F}, \pcompo^\mathcal{F}, {\pconve}^\mathcal{F}, \pid^\mathcal{F}, \pfork^\mathcal{F} \>$ and $\mathcal{G} = \< G, \pjoin^\mathcal{G}, \pmeet^\mathcal{G}, \pcomple^\mathcal{G}, \pzero^\mathcal{G}, \punit{E}^\mathcal{G}, \pcompo^\mathcal{G}, {\pconve}^\mathcal{G}, \pid^\mathcal{G}, \pfork^\mathcal{G} \>$ be $\PFA$ and $s \in \Secu$, if $\phi: \mathcal{F} \to \mathcal{G}$ is an isomorphism, then $\phi \(\underline{s}^\mathcal{F}\) = \underline{s}^\mathcal{G}$.
\end{lemma}
\begin{proof}
The proof of this lemma follows by induction on the structure of $s$. Let us consider $s = \elem{\pi}$ as the base case. The case for $s = \elem{\rho}$ is analogous.
\[
\begin{array}{rclr}
\phi \(\underline{\elem{\pi}}^\mathcal{F}\) & = & \phi \({\pi}^\mathcal{F}\) & \text{[by Def.~\ref{sec-underline}.]}\\
		& = & \phi \(  {\pconv{\pid^\mathcal{F} \pfork^\mathcal{F} \punit{E}^\mathcal{F}}}^\mathcal{F} \) & \text{[by Def.~$\pi$.]}\\
		& = & {\pconv{\pid^\mathcal{G} \pfork^\mathcal{G} \punit{E}^\mathcal{G}}}^\mathcal{G} & \text{[because $\phi$ is isomorphism.]}\\
		& = & \pi^\mathcal{G} & \text{[by Def.~$\pi$.]}\\
		& = & \underline{\elem{\pi}}^\mathcal{G} & \text{[by Def.~\ref{sec-underline}.]}
\end{array}
\]
Let us now consider the case $s = \cons{\pi}{s'}$. As in the previous case, the case for $s = \cons{\rho}{s'}$ is analogous.
\[
\begin{array}{rcll}
\phi \(\underline{\cons{\pi}{s'}}^\mathcal{F}\) & = & \phi \( \underline{\pi}^\mathcal{F} \pcompo^\mathcal{F} \underline{s'}^\mathcal{F} \) & \text{[by Def.~\ref{sec-underline}.]}\\
		& = & \underline{\pi}^\mathcal{G} \pcompo^\mathcal{G} \phi \( \underline{s'}^\mathcal{F} \) & \text{[because $\phi$ is isomorphism.]}\\
		& = & \underline{\pi}^\mathcal{G} \pcompo^\mathcal{G} \underline{s'}^\mathcal{G} & \text{[by inductive hypothesis.]}\\
		& = & \underline{\cons{\pi}{s'}}^\mathcal{G} & \text{[by Def.~\ref{sec-underline}.]}
\end{array}
\]
Thus, finishing the proof.\qed\end{proof}

\begin{lemma}
\label{Sfixpoint-s}
Let $U$ and $S \in U$ be a non-empty sets, $s \in \Secu$ and $\star_S: U^2 \to U$ be the injective function defined according to Prop.~\ref{existsEXPsconst}, then $u \in S$ implies $\langle u, u \rangle = \underline{s}$. 
\end{lemma}
\begin{proof}
Let us assume that $s[1] = \pi$, the case in which $s[1] = \rho$ follows analogously.
\[
\begin{array}{cl}
 & u \in S \\
\text{iff}          & \text{there exists } v \in S_{\mathit{long} (s) - 1} \text{ such that } u = {h^s_{\mathit{long} (s) - 1}}^{-1} (v)\\
 & \qquad\qquad\qquad\qquad\qquad\qquad \text{[because $h^s_{\mathit{long} (s) - 1}$ is a bijection.]}\\ 
\text{iff}          & \text{there exists } v \in S_{\mathit{long} (s) - 1} \text{ such that } u = {h^s_{\mathit{long} (s) - 1}}^{-1} (v) \text{ and }\\
 & \qquad u = v \star_S l \( {h^s_{\mathit{long} (s) - 1}}^{-1} (v) \)\\
 & \qquad\qquad\qquad\qquad\qquad\qquad \text{[by Def.~$\star_S$.]}\\ 
 \text{iff}          & \text{there exists } v \in S_{\mathit{long} (s) - 1} \text{ such that } u = {h^s_{\mathit{long} (s) - 1}}^{-1} (v) \text{ and }\\
 & \qquad \langle u, v \rangle \in \underline{\pi} \\
 & \qquad\qquad\qquad\qquad\qquad\qquad \text{[by Def.~\ref{underlinepi}.]}\\ 
 \text{implies} & \text{there exists } v \in S_{\mathit{long} (s) - 1} \text{ such that } u = {h^s_{\mathit{long} (s) - 1}}^{-1} (v) \text{ and }\\
 & \qquad \langle u, v \rangle \in \underline{\pi} \text{ and } \< v, {h^s_ {\mathit{long} (s) - 1}}^{-1} (v)\> \in s | \mathit{long} (s) - 1\\
 & \qquad\qquad\qquad\qquad\qquad\qquad \text{[by Lemma~\ref{aux-Sfixpoint-s}.]}\\ 
 \text{iff}          & \text{there exists } v \in S_{\mathit{long} (s) - 1} \text{ such that } \langle u, v \rangle \in \underline{\pi} \text{ and } \\
 & \qquad  \langle v, u \rangle \in s | \mathit{long} (s) - 1 \text{ and } \< u, u\> \in \underline{\pi} \pcompo s | \mathit{long} (s) - 1\\
 & \qquad\qquad\qquad\qquad\qquad\qquad \text{[by Def.~\ref{def_proper-relation-algebra} - $\pcompo$.]}\\ 
 \text{iff}          & \< u, u\> \in \underline{\pi} \pcompo s | \mathit{long} (s)\\
 & \qquad\qquad\qquad\qquad\qquad\qquad \text{[because $s[1] = \pi$.]}\\ 
 \text{iff}          & \< u, u\> \in \underline{s}\\
 & \qquad\qquad\qquad\qquad\qquad\qquad \text{[by Def.~\ref{functions} - $s | \mathit{long} (s)$.]}           
\end{array}
\]
Thus, finishing the proof.\qed\end{proof}

\begin{lemma}
\label{aux-Sfixpoint-s}
Let $U$ and $S \in U$ be a non-empty sets, $s \in \Secu$ and $\star_S: U^2 \to U$ be the injective function defined according to Prop.~\ref{existsEXPsconst}, then for all $i \in \nat$ such that $1 \leq i < \mathit{long} (s)$, $v \in S_i$ implies $\<v, {h_i^s}^{-1} (v)\> \in s | i$. 
\end{lemma}
\begin{proof}
The proof follows by induction on $i$. Let $i = 1$ such that $1 < \mathit{long} (s)$ and assuming that $s [\mathit{long} (s) - 1] = \pi$. The case where $s [\mathit{long} (s) - 1] = \rho$ is analogous.
\[
\begin{array}{rcl}
v \in S_1 & \text{implies} & \text{there exists } u \in S \text{ such that } v = h_1^s (u)\\
               &                      & \qquad\qquad\qquad\quad \text{[because $h^s_1$ is a bijection.]}\\ 
               & \text{iff}         & \text{there exists } u \in S \text{ such that } v = h_1^s (u) \text{ and } \\
               &                      & \qquad v = h_1^s \({h_0^s}^{-1} (u)\)\\
               &                      & \qquad\qquad\qquad\quad \text{[by Def.~ $h^s_0$.]}\\ 
               & \text{iff}         & \text{there exists } u \in S \text{ such that } v = h_1^s (u) \text{ and } \\
               &                      & \qquad v = u \star_S l \({h_0^s}^{-1} (u)\)\\
               &                      & \qquad\qquad\qquad\quad \text{[by Def.~ $\star_S$.]}\\ 
               & \text{iff}         & \text{there exists } u \in S \text{ such that } {h_1^s}^{-1} (v) = u \text{ and } \\
               &                      & \qquad v = u \star_S l \({h_0^s}^{-1} (u)\)\\
               &                      & \qquad\qquad\qquad\quad \text{[because $h^s_1$ is a bijection.]}\\ 
               & \text{implies} & \text{there exists } u \in S \text{ such that } {h_1^s}^{-1} (v) = u \text{ and } \\
               &                      & \qquad \<v, u\> \in \underline{\pi}\\
               &                      & \qquad\qquad\qquad\quad \text{[by Def.~\ref{underlinepi}.]}\\ 
               & \text{iff}         & \<v, {h_1^s}^{-1} (v)\> \in \underline{\pi}\\
               & \text{iff}         & \<v, {h_1^s}^{-1} (v)\> \in \underline{s |1}\\
               &                      & \qquad\qquad\qquad\quad \text{[by Def.~\ref{functions} - $s | 1$ and $s[\mathit{long} (s) - 1] = \pi$.]} 
\end{array}
\]
Let $i = n+1$ such that $1 \leq n+1 < \mathit{long} (s)$ and assuming that $s [\mathit{long} (s) - (n+1)] = \pi$. The case where $s [\mathit{long} (s) - (n+1)] = \rho$ is analogous.
\[
\begin{array}{rcl}
v \in S_{n + 1} & \text{implies} & \text{there exists } u \in S \text{ such that } v = h_{n + 1}^s (u)\\
                       &                      & \qquad\qquad\qquad\qquad\qquad \text{[because $h^s_{n + 1}$ is a bijection.]}\\ 
                       & \text{iff}         & \text{there exists } u \in S \text{ such that } v = h_{n + 1}^s (u) \text{ and } \\
                       &                      & \qquad \text{there exists } u' \in S \text{ such that } u = {h_{n}^s}^{-1} (u')\\
                       &                      & \qquad\qquad\qquad\qquad\qquad \text{[because $h_{n}^s$ is a bijection.]}\\ 
                       & \text{iff}         & \text{there exists } u' \in S \text{ such that } v = h_{n + 1}^s \({h_{n}^s}^{-1} (u')\)\\
                       &                      & \qquad\qquad\qquad\qquad\qquad \text{[because $h_{n}^s$ is a bijection.]}\\ 
                       & \text{iff}         & \text{there exists } u' \in S \text{ such that } v = h_{n + 1}^s \({h_{n}^s}^{-1} (u')\) \text{ and } \\
                       &                      & \qquad v = u'  \star_S l \({h_{n}^s}^{-1} (u')\)\\
                       &                      & \qquad\qquad\qquad\qquad\qquad \text{[by Def.~$\star_S$.]}\\ 
                       & \text{iff}         & \text{there exists } u' \in S \text{ such that } v = h_{n + 1}^s \({h_{n}^s}^{-1} (u')\) \text{ and } \\
                       &                      & \qquad v = u'  \star_S l \({h_{n}^s}^{-1} (u')\) \text{ and } \<u', {h_i^s}^{-1} (u')\> \in \underline{s | n}\\
                       &                      & \qquad\qquad\qquad\qquad\qquad \text{[by inductive hypothesis.]}\\ 
                       & \text{iff}         & \text{there exists } u' \in S \text{ such that } {h_{n + 1}^s}^{-1} (v) = {h_{n}^s}^{-1} (u') \text{ and } \\
                       &                      & \qquad v = u'  \star_S l \({h_{n}^s}^{-1} (u')\) \text{ and } \<u', {h_i^s}^{-1} (u')\> \in \underline{s | n}\\
                       &                      & \qquad\qquad\qquad\qquad\qquad \text{[because $h_{n + 1}^s$ is a bijection.]}\\ 
                       & \text{iff}         & \text{there exists } u' \in S \text{ such that } v = u'  \star_S l \({h_{n}^s}^{-1} (u')\) \text{ and } \\
                       &                      & \qquad \<u', {h_{n + 1}^s}^{-1} (v)\> \in \underline{s | n}\\
                       & \text{implies} & \text{there exists } u' \in S \text{ such that } \< v, u' \> \in \underline{\pi} \text{ and } \\
                       &                      & \qquad \<u', {h_{n + 1}^s}^{-1} (v)\> \in \underline{s | n}\\
                       &                      & \qquad\qquad\qquad\qquad\qquad \text{[by Def.~\ref{underlinepi}.]}\\ 
                       & \text{iff}         & \< v, {h_{n + 1}^s}^{-1} (v)\> \in \underline{\pi} \pcompo \underline{s | n}\\
                       &                      & \qquad\qquad\qquad\qquad\qquad \text{[by Def.~\ref{def_proper-relation-algebra} - $\pcompo$.]}\\ 
                       & \text{iff}         & \< v, {h_{n + 1}^s}^{-1} (v)\> \in \underline{s | (n + 1)}\\
                       &                      & \qquad\qquad\qquad\qquad\qquad \text{[by Def.~\ref{functions} - $s | (n + 1)$ and }\\
                       &                      & \qquad\qquad\qquad\qquad\qquad\qquad \text{$s[\mathit{long} (s) - (n + 1)] = \pi$.]} 
\end{array}
\]
Thus, finishing the proof.\qed\end{proof}

\begin{lemma}
\label{fix-inclusion}
Let $s, s' \in \Secu$, $\underline{s} \pcompo \underline{s'} = \underline{\concat{s}{s'}}$.
\end{lemma}
\begin{proof}
The proof follows easily by induction on the structure of $s$.
\qed\end{proof}

\begin{lemma}
\label{fix-t-fix-s}
Let $U$ be an infinite set $|U| = \kappa$ with $\aleph_0 \leq \kappa$, $S \subseteq U$ such that $|S| < |U|$, $\star_S: U^2 \to U$ injective and $t \in \BT$ then, for all $s \in \Secu$ such that $s \ll t$, for all $u, v \in U$ if $u = \map{t}{\star_S}{v}$ then $\<u, v\> \in \underline{s}$.
\end{lemma}
\begin{proof}
The proof follows by induction on the structure of $s$. Let us consider $s = \elem{\pi}$ as the base case. The case for $s = \elem{\rho}$ is analogous. Assume that $s \ll t$ and, consequently, $t = \bin{\nil}{t'}$ with $t' \in \BT$.
\[
\begin{array}{rcl}
u = \map{\(\bin{\nil}{t'}\)}{\star_S}{v} & \text{iff}          & u = v \star_S \(\map{t'}{\star_S}{v}\) \\
                                                         &                      & \qquad\qquad \text{[by Def.~\ref{map}.]}\\
                                                         & \text{implies} & \text{there exists } w \in U \text{ such that } \\
                                                         &                      & u = v \star_S w \text{ and } \< u, v \> \underline{\pi}\\
                                                         &                      & \qquad\qquad \text{[by Def.~\ref{underlinepi}.]}
\end{array}
\]
Let us now consider the case $s = \cons{\pi}{s'}$ and assuming $s \ll t$, $t = \bin{t'}{t''}$ such that $t' \neq \nil$. As in the previous case, the case for $s = \cons{\rho}{s'}$ is analogous.
\[
\begin{array}{rcl}
u = \map{\(\bin{t'}{t''}\)}{\star_S}{v}   & \text{iff}          & u = \(\map{t'}{\star_S}{v}\) \star_S \(\map{t''}{\star_S}{v}\) \\
                                                         &                      & \qquad\qquad \text{[by Def.~\ref{map}.]}\\
                                                         & \text{implies} & \text{there exists } w, w' \in U \text{ such that } \\
                                                         &                       & \qquad u = w' \star_S w \text{ and } w' = \map{t'}{\star_S}{v}\\
                                                         & \text{iff}          & \text{there exists } w, w' \in U \text{ such that } \\
                                                         &                       & \qquad u = w' \star_S w \text{ and } \< w', v \> \in \underline{s'}\\
                                                         &                       & \qquad\qquad \text{[by inductive hypothesis on}\\
                                                         &                       & \qquad\qquad\qquad \text{ $s'$ and $s \ll t$.]}\\
                                                         & \text{iff}          & \text{there exists } w' \in U \text{ such that } \\
                                                         &                       & \qquad \< u, w' \> \in \underline{\pi} \text{ and } \< w', v \> \in \underline{s'}\\
                                                         &                       & \qquad\qquad \text{[by Def.~\ref{underlinepi}.]}\\
                                                         & \text{iff}          & \< u, v \> \in \underline{\pi} \pcompo \underline{s'}\\
                                                         &                       & \qquad\qquad \text{[by Def.~\ref{def_proper-relation-algebra} - $\pcompo$.]}\\
                                                         & \text{iff}          & \< u, v \> \in \underline{\cons{\pi}{s'}}\\
                                                         &                       & \qquad\qquad \text{[by Def.~\ref{sec-underline}.]}\\
                                                         & \text{iff}          & \< u, v \> \in \underline{s}
\end{array}
\]
Thus, finishing the proof.\qed\end{proof}

\end{document}